\documentclass[twoside,leqno,twocolumn]{article}

\makeatletter
\gdef\@copyrightpermission{
\begin{minipage}{0.3\columnwidth}
 \href{https://creativecommons.org/licenses/by/4.0/}{\includegraphics[width=0.90\textwidth]{figures/4ACM-CC-by-88x31.eps}}
\end{minipage}\hfill
\begin{minipage}{0.7\columnwidth}
 \href{https://creativecommons.org/licenses/by/4.0/}{This work is licensed under a Creative Commons Attribution International 4.0 License.}
\end{minipage}
\vspace{5pt}
}
\makeatother

\usepackage[T1]{fontenc}
\usepackage{lmodern}
\RequirePackage[tt=false, type1=true]{libertine}
\RequirePackage[varqu]{zi4}
\RequirePackage[libertine]{newtxmath}
\RequirePackage[T1]{fontenc}

\usepackage{booktabs}   
\usepackage{subcaption} 
\usepackage{caption}
\usepackage{colortbl}
\usepackage{bigstrut}
\usepackage{microtype}
\usepackage{tabularx}
\usepackage{multirow}
\usepackage{multicol}
\usepackage{rotating}
\usepackage{lipsum}
\usepackage{balance}
\usepackage{mfirstuc}
\usepackage{titlecaps}
\usepackage{listings}
\usepackage{float}
\usepackage[rightcaption]{sidecap}
\usepackage{xcolor}
\usepackage{amsmath}
\usepackage[font=small,labelfont=bf,skip=2pt]{caption}
\usepackage[a-2b]{pdfx}
\usepackage{ltexpprt}

\usepackage{cleveref}
\crefname{section}{Sec.}{Sec.}
\crefname{theorem}{Thm.}{Thm.}
\crefname{lemma}{Lemma}{Lemma}
\crefname{corollary}{Cor.}{Cor.}
\crefname{table}{Tab.}{Tab.}
\crefname{algorithm}{Alg.}{Alg.}
\crefname{definition}{Definition}{Definition}
\crefname{figure}{Fig.}{Fig.}
\crefname{fact}{Fact}{Fact}
\Crefname{table}{Tab.}{Tab.}
\crefname{problem}{Problem}{Problem}
\crefname{line}{Line}{Line}

\def\fullversion{}

\usepackage{mfirstuc}

\usepackage{etoolbox}
\newcommand{\ifconference}[1]{{{\ifx\fullversion\undefined{#1}\fi}\xspace}}
\newcommand{\iffullversion}[1]{{{\ifx\conference\undefined{#1}\fi}\xspace}}

\usepackage{graphicx}  
\usepackage{lipsum}  
\newcommand{\hide}[1]{} 
\usepackage{xspace}
\usepackage{textcomp}
\usepackage{comment} 
\usepackage{verbatim}

\usepackage[numbers,sort,compress]{natbib}
\usepackage{url}

\newcommand{\emp}[1]{{\boldmath\emph{\textbf{#1}}\unboldmath}} 
\newcommand{\fname}[1]{\textsf{#1}} 
\newcommand{\mathfunc}[1]{\mathit{#1}}

\usepackage[lmargin=0.7in,rmargin=0.7in,tmargin=.85in,bmargin=.85in]{geometry}

\newtheorem{definition}{Definition}

\let \originalleft \left
\let\originalright\right
\renewcommand{\left}{\mathopen{}\mathclose\bgroup\originalleft}
\renewcommand{\right}{\aftergroup\egroup\originalright}

\usepackage{scalerel} 

\newcommand{\whp}[1]{\emph{whp}}

\usepackage{pifont}

\usepackage[shortlabels]{enumitem}

\setlist{topsep=0.3em,itemsep=0.2em,parsep=0.1em,leftmargin=*}

\usepackage{float}
\usepackage[font={small},aboveskip=0.3em, belowskip=0.2em]{caption}

\usepackage[labelfont=bf,list=true,skip=0em]{subcaption}
\usepackage[rightcaption]{sidecap}

\usepackage{wrapfig}

\usepackage{array}
\newcolumntype{L}[1]{>{\raggedright\let\newline\\\arraybackslash\hspace{0pt}}m{#1}}
\newcolumntype{C}[1]{>{\centering\let\newline\\\arraybackslash\hspace{0pt}}m{#1}}
\newcolumntype{R}[1]{>{\raggedleft\let\newline\\\arraybackslash\hspace{0pt}}m{#1}}
\newcolumntype{B}{>{\bf}c}
\usepackage{rotating}

\usepackage{booktabs} 
\usepackage{multicol,multirow}
\usepackage{longtable} 
\usepackage{supertabular} 
\usepackage{colortbl}
\usepackage{bigstrut}

\newcommand{\myparagraph}[1]{\noindent\emp{#1}~~}

\usepackage[ruled,lined,linesnumbered,noend]{algorithm2e}
\usepackage[noend]{algpseudocode}

\makeatletter
\patchcmd{\@algocf@start}
  {-1.5em}
  {0pt}
  {}{}
\newcommand{\nosemic}{\renewcommand{\@endalgocfline}{\relax}}
\newcommand{\dosemic}{\renewcommand{\@endalgocfline}{\algocf@endline}}

\SetSideCommentLeft
\SetKwInput{notations}{Notations}
\SetKwInput{notes}{Notes}
\SetKwInput{maintains}{Maintains}

\SetKwProg{myfunc}{Function}{}{}
\SetKwFor{parForEach}{ParallelForEach}{do}{endfor}
\SetKwFor{Justrepeat}{Repeat}{}{}

\SetKw{MIN}{min}
\SetKw{MAX}{max}
\SetKw{OR}{or}
\SetKw{AND}{and}

\SetCommentSty{mycommfont}

\usepackage{mdframed}
\mdfdefinestyle{mystyle}{innertopmargin=2pt,innerbottommargin=2pt,skipabove=2pt,skipbelow=2pt}
\mdfdefinestyle{densestyle}{linecolor=framelinecolor,innertopmargin=0,innerbottommargin=0,leftmargin=0,rightmargin=0,backgroundcolor=gray!20}
\mdfdefinestyle{compactcode}{linecolor=framelinecolor,innertopmargin=1pt,innerbottommargin=1pt,backgroundcolor=gray!20,skipabove=0pt,skipbelow=0pt,leftmargin=0,rightmargin=0}

\usepackage{framed}

\usepackage{listings}

\newdimen\zzsize
\newdimen\kwsize
\newcommand{\basicstyle}{\fontsize{\zzsize}{1\zzsize}\ttfamily}
\newcommand{\keywordstyle}{\fontsize{\kwsize}{1\kwsize}\ttfamily\bf}

\newdimen\zzlstwidth
\usepackage{tikz} 

\hide{
\binoppenalty=700
\brokenpenalty=0 
\clubpenalty=0   
\displaywidowpenalty=0   
\exhyphenpenalty=50
\floatingpenalty=20000
\hyphenpenalty=50
\interlinepenalty=0
\linepenalty=10
\postdisplaypenalty=0
\predisplaypenalty=0 
\relpenalty=500
\widowpenalty=0  
}

\newcommand{\yihan}[1]{{\color{purple}{\bf Yihan:} #1}}

\newcommand{\xiangyun}[1]{{\color{blue}{\bf Xiangyun:} #1}}

\newcommand{\op}{\mathit{op}}
\newcommand{\amortized}[1]{\mathit{C}_{\mathit{amortized}}(#1)}
\newcommand{\actual}[1]{\mathit{C}_{\mathit{actual}}(#1)}

\newcommand{\ourtreefull}{Anti-Monopoly tree}
\newcommand{\ourtree}{AM-tree}

\newcommand{\tmstfull}{transformed MST}
\newcommand{\tmst}{T-MST}

\newcommand{\ourforest}{T-MSF}

\newcommand{\funcfont}[1]{{\textsf{#1}}}

\newcommand{\pmeq}{PM-equivalent}
\newcommand{\parent}[1]{\mathfunc{parent}[#1]}
\newcommand{\parentt}[2]{\mathfunc{parent}_{#1}[#2]}
\newcommand{\weight}[1]{\mathfunc{weight}[#1]}
\newcommand{\size}[1]{\mathfunc{size}[#1]}

\newcommand{\zigzag}{rotate}
\newcommand{\shortcut}{shortcut}
\newcommand{\perch}{perch}
\newcommand{\perchfunc}{\fname{Perch}}
\newcommand{\stitch}{stitch}
\newcommand{\stitchfunc}{\fname{Stitch}}

\newcommand{\promote}{promote}
\newcommand{\promoted}{promoted}
\newcommand{\promotefunc}{\fname{Promote}}
\newcommand{\link}{link}
\newcommand{\linkfunc}{\fname{Link}}

\newcommand{\calibrate}{calibrate}
\newcommand{\calibratefunc}{\fname{Calibrate}}

\newcommand{\downwardcalifunc}{\fname{DownwardCalibrate}}
\newcommand{\upwardcalifunc}{\fname{UpwardCalibrate}}
\newcommand{\downwardmaintain}{\downwardcalifunc}
\newcommand{\upwardmaintain}{\upwardcalifunc}
\newcommand{\pathmax}{\fname{PathMax}}
\newcommand{\lazypathmax}{\fname{LazyPathMax}}
\newcommand{\perchandlink}{\fname{\titlecap{\link}By\titlecap{\perch}}}
\newcommand{\stitchandlink}{\fname{\titlecap{\link}By\titlecap{\stitch}}}

\newcommand{\insertfunc}{\fname{Insert}}
\newcommand{\lazyinsertfunc}{\fname{LazyInsert}}
\newcommand{\original}[1]{\hat{#1}}

\begin{document}

\title{
New Algorithms for Incremental Minimum Spanning Trees and \\ Temporal Graph Applications
}
\date{}

\author{Xiangyun Ding\\UC Riverside\\xding047@ucr.edu
    \and Yan Gu\\UC Riverside\\ygu@cs.ucr.edu
    \and Yihan Sun\\UC Riverside\\yihans@cs.ucr.edu
}

\setlength\abovedisplayskip{0.2em}
\setlength\belowdisplayskip{0.2em}
\setlength\abovedisplayshortskip{0.2em}
\setlength\belowdisplayshortskip{0.2em}

\maketitle


\fancyfoot[R]{\scriptsize{Copyright \textcopyright\ 2025 by SIAM\\
Unauthorized reproduction of this article is prohibited}}





\begin{abstract}

  Processing graphs with temporal information (the \emph{temporal graphs}) has become increasingly important in the real world.
  In this paper, we study efficient solutions to temporal graph applications using new algorithms for \emph{Incremental Minimum Spanning Trees} (MST). 
  The first contribution of this work is to formally discuss how a broad set of setting-problem combinations of temporal graph processing can be solved
  using incremental MST, along with their theoretical guarantees. 
  
  Despite the importance of the problem, we observe a gap between theory and practice for efficient incremental MST algorithms. 
  While many classic data structures, such as the link-cut tree, provide strong bounds for incremental MST, their performance is limited in practice. 
  Meanwhile, existing practical solutions used in applications do not have any non-trivial theoretical guarantees. 
  Our second and main contribution includes new algorithms for incremental MST that are efficient both in theory and in practice. 
  Our new data structure, the \emph{\ourtree{}}, achieves the same theoretical bound as the link-cut tree for temporal graph processing and shows strong performance in practice.
  In our experiments, the \ourtree{} has competitive or better performance than existing practical solutions due to theoretical guarantees, 
  and can be significantly faster than the link-cut tree (7.8--11$\times$ in updates and 7.7--13.7$\times$ in queries). 
\end{abstract}

\section{Introduction}

The concept of graphs is vital in computer science.
It is relevant to lots of applications as it abstracts real-world objects as vertices and their relationships as edges.
Regarding the relationships between objects, time can usually be a crucial component.
Graphs with time information are referred to as \emph{temporal graphs}, and efficient algorithms for temporal graphs have received immense attention recently.
Time information can be integrated in different settings.
A classic setting is that each edge has a timestamp, and a query, such as connectivity, is augmented with a time interval $[t_1,t_2]$, and only edges within this time period are involved in the query.
Dually, each edge $e$ can have a time period $[t_1, t_2]$; a query is on a certain timestamp $t$, and only considers edges existing at time $t$. 
Meanwhile, edges and queries can come in either offline (known ahead of time) or online (immediate response needed) manner.
Combined with numerous graph problems, there are a large number of research topics
(a short list of papers in the recent years:~\cite{huber2022streaming,gandhi2020interval,ciaperoni2020relevance,li2018persistent,anderson2020work,crouch2013dynamic,song2024querying,tian2024efficient,xie2023querying,yang2023scalable,yu2021querying,yang2024evolution,zhang2024incremental,da2024kairos,peng2019optimal,holme2013epidemiologically,rocha2013bursts,bearman2004chains}).
Most of them focus on one specific setting-problem combination. 

In this paper, we are interested in solutions for a class of temporal graph applications for a wide range of setting-problem combinations,
both in theory and in practice. 
Our core algorithmic idea is to support an efficient data structure for the \emph{incremental minimum spanning trees (MST)}. 
The MST for a weighted undirected graph $G=(V,E)$
is a subgraph $T = (V, E')$ such that $E' \subseteq E$ and $T$ is a tree that connects all vertices in $V$ with minimum total edge weight.
The incremental MST problem requires maintaining the MST while responding to edge insertions. 
Some existing studies~\cite{anderson2020work,crouch2013dynamic,song2024querying}, both from the algorithm and application communities, have shown connections between incremental MST to a list of specific temporal graph applications. 
At a high level, one can embed the temporal information into the edge weight, and temporal queries can then be converted to \emp{path-max} queries on the MST, i.e., reporting the maximum edge weight on the path between two queried nodes. 
We show a running example in \cref{sec:prelim:temporal}. 
\emph{The first contribution of this paper is to \textbf{formally discuss (in \cref{sec:app}) a wide range of temporal graph applications with different setting-problem combinations, and how incremental MST can be adapted to address them}. 
} 


\hide{
The concept of graphs are vital in computer science and used in lots of applications.
Beyond processing a single, static graph, many real-world applications involve temporal information associated with graph edges.
This has led to increased research interest in temporal graph processing~\cite{}
Driven by the diverse range of applications, various types of temporal graphs have been explored in the literature.
One common setting is to associate each edge with a single timestamp indicating when the edge becomes valid.
A query is augmented with a time interval $[s,t]$, such that only edges within this time period are considered in the graph.
Dually, temporal information for each edge $e$ can also be a time period $[s_e, t_e]$, indicating that the edge exists in the graph from time $s_e$ to $t_e$.
A query then asks to analyze the graph at a certain timestamp $x$, such that only edges valid at time $x$ are included in the graph.
These temporal graph problems arise naturally in a wide range of real-world applications, such as xx, xx, xx.
}

\newcommand{\versionone}{
\textbf{Version 1:}}

\newcommand{\versiontwo}{
\textbf{Version 2:}

Computing the Minimum Spanning Tree (MST) of is a classic fundamental graph processing problem.
Given a (connected) weighted undirected graph $G=(V,E)$,
a spanning tree is a subset of edges in $E$ containing exactly $n-1$ edges that forms a tree structure connecting all vertices in the graph, where $n=|V|$ is the number of vertices.
The minimum spanning tree is the spanning tree with the minimum total edge weight.
In the more general case where the graph may be disconnected, the problem extends to computing the minimum spanning forest, which consists of the minimum spanning trees for each connected component of the graph.

Beyond being an fundamental problem on its own, many important problems are related to the \emp{incremental MST}, which involves maintaining the MST while responding to edges insertions in an online manner.
In this paper, we are interested in the applications of incremental MST on \emp{temporal graphs}.
Temporal graphs involve additional temporal information associated to edges,
and such settings are commonly found in modern real-world applications.
In modern graph processing applications, edges are often associated with temporal information.
This need has propelled temporal graph processing into a prominent research area, as evidenced by recent studies~\cite{}.
Driven by the diverse range of applications, existing literature explores various types of temporal graphs.
One common setting is to associate each edge with a single timestamp indicating when the edge becomes valid.
A query is augmented with a time interval $[s,t]$, such that only edges within this time period are considered in the graph.
Dually, temporal information for each edge $e$ can also be a time period $[s_e, t_e]$, during which it exists in the graph.
A query then focuses on a certain timestamp $x$, such that only edges exist at time $x$ are included in the graph.
These temporal graph problems arise naturally in a wide range of real-world applications, such as xx, xx, xx.

Many existing studies have revealed the connections between incremental MST and temporal graph processing.
As a simple example, consider a temporal graph such that each edge $e$ has a validity period $[s_e,t_e]$,
and each query asks whether two vertices $u$ and $v$ are connected at time $x$.
In this setting, Eppstein's work~\cite{} implies an MST-based approach,
where an MST can be built based on incrementally inserting each edge $e$ at time $s_e$ with weight $-t_e$.
A query at time $x$ will first obtain the MST $T$ at time $x$, and compute the maximum edge weight between $u$ and $v$ on $T$ as $w$.
If $|w|\le x$, it means that to connect $u$ and $v$, an edge with weight smaller than $x$ (i.e., it has been deleted before time $x$) is necessary, and therefore $u$ and $v$ are not connected at time $x$. Otherwise, $u$ and $v$ are connected by the tree path on $T$ where all edges still exist at time $x$.
The dual version of the problem and its reduction to incremental MST has been shown by Song et al.~\cite{xx}.
The first contribution of this paper is to extend the idea and more formally discuss 20? different temporal graph processing applications that can benefit from incremental MST. These problems include connectivity, $k$-connectivity, xx, xx, with different combinations of online and offline queries.

Given the importance, incremental MST and related problems have been widely studied.
Many classic data structures can support incremental MST.
For example, the famous link-cut tree~\cite{sleator1983data} can maintain the incremental MST with $O(\log n)$ time per insertion, and a path-max query in $O(\log n)$ time, both amortized.
The rake compress tree~\cite{anderson2024deterministic,acar2005experimental} provides the same bound in the worst case.
Other relevant data structures include Top Tree \cite{tarjan2005self}, Euler-tour Tree \cite{henzinger1999randomized}, etc.
Despite the strong bounds in theory, these results are often considered to have limited practicality in performance and/or programmability.
In 2024, Song et al. proposed a data structure called OEC-forest~\cite{song2024querying}, and tested it against link-cut tree in the application of temporal graph.
In various experiments, OEC-forest show xx-xx speedup than link-cut tree.
The OEC tree, on the other hand, has no theoretical guarantee for the incremental MST problem.
This motivates us to consider whether there exists a simple solution for incremental MST that is efficient both in theory and in practice.

\emp{The second and the main contribution of this paper is a new, theoretically and practically-efficient algorithm for incremental MST.}
}




Given the broad applicability, efficient incremental MST algorithms are of great importance.  
Indeed, many classic data structures provide efficient solutions in theory.
For example, the famous link-cut tree~\cite{sleator1983data} can maintain the incremental MST with $O(\log n)$ time per insertion, and a path-max query in $O(\log n)$ time, both amortized.
Other relevant data structures (e.g., the rake-compress tree (RC-tree)~\cite{acar2005experimental} and the top tree~\cite{tarjan2005self}) can provide similar bounds.
Despite the strong bounds in theory, these results are often considered to have limited practicality due to large hidden constants and/or high programming complexity.
Many other data structures, such as OEC-forest~\cite{song2024querying} and D-tree~\cite{chen2022dynamic}, are used in practice and can be much faster than the link-cut tree. 
Experiments in~\cite{song2024querying} show that, on a specific temporal graph processing application,
the OEC-forest is up to 15$\times$ faster than the link-cut tree in updates and 13$\times$ in queries. 
However, no non-trivial bounds (better than $O(n)$ per operation) are known for these practical data structures. 
Hence, it remains open whether an efficient solution exists for incremental MST (and relevant temporal graph applications) \emp{both in theory and in practice}.

\emph{The second and the main contribution of this paper is \textbf{a new, theoretically and practically efficient data structure for incremental MST, referred to as the \ourtreefull{} (the \ourtree{})}}.
In addition to strong theoretical guarantees and practical efficiency, 
the algorithms of \ourtree{} are also simple, 
leading to good programmability and applicability to real-world problems. 
An \ourtree{} $T$ is a rooted tree that reflects a transformation of the MST $\original{T}$ of the graph, such that for any two vertices $u$ and $v$, the path-max query on $T$ is the same as in $\original{T}$. 
The most important property of \ourtree{} is the \emph{anti-monopoly rule (AM-rule)}, 
which requires each subtree size to be no more than a factor of 2/3 of its parent. 
This ensures $O(\log n)$ tree height for a tree with size $n$, and thus bounded cost for updating and searching the tree. 
The algorithm for \ourtree{s} is based on two simple primitives. 
The first primitive, \linkfunc{}$(u,v,w)$, incorporates a new edge between $u$ and $v$ with weight $w$ inserted to the original graph. 
\linkfunc{} will properly update the tree to ensure that \ourtree{} 
still preserves the correct answers to path-max queries to the new graph, but may violate the size constraint of the tree. 
The second primitive, \calibratefunc{}, modifies the tree to obey the AM-rule, and thus restores the logarithmic tree height. 
In \cref{sec:preemptive}, we first present algorithms that strictly keep the tree height in $O(\log n)$ after handling edge insertions, which we call the \emp{strict \ourtree{}}. 
We provide two algorithms for \linkfunc{}: \perchandlink{}, which is algorithmically simpler, and \stitchandlink{}, which performs better in practice.
In both cases, we prove that a path-max query can be performed in $O(\log n)$ worst-case cost, and each insertion can be performed with $O(\log n)$ amortized cost ($O(\log^2 n)$ in the worst case).
The theoretical results are presented in Thm. \ref{thm:exact}. 

The strict \ourtree{}, however, requires maintaining the child pointers in each node, which may increase performance overhead in practice.
In \cref{sec:lazy}, we further extend \ourtree{} to the \emp{lazy \ourtree{}}, which does not rebalance the tree immediately, but postpones the \calibratefunc{} operation to the next time when a node is accessed.
The lazy version directly uses the same \link{} primitive as the strict version, which can be either \perchfunc{}-based or \stitchfunc{}-based.
It redesigns \calibratefunc{} such that it can be performed lazily, and only requires each node to maintain the parent pointer.
Compared to the strict version, the lazy version achieves the same $O(\log n)$ amortized cost for insertion and path-max query, 
and provides better performance in practice. 

For all versions of \ourtree{}, the (amortized) theoretical bounds match the best-known bounds of link-cut tree. 
The core idea to achieve the bounds is based on the potential function in Eq. \ref{eq:potential}, 
such that the AM rule can be incorporated to ensure the potential does not increase much during updates, and can always be restored by the \calibratefunc{} functions. 

To support more settings in temporal graph processing, we also persist \ourtree{s} in \cref{sec:persistent}.
A persistent data structure keeps all history versions of itself upon updates.
Our solution is based on a standard approach using version lists~\cite{reed1978naming,persistence},
which preserves the same asymptotic cost for insertions and incurs a logarithmic overhead per path-max query.

Using \ourtree{} to support incremental MST, we can derive solutions for various temporal graph processing.
In \cref{sec:app}, we discuss a series of relevant applications and their solutions using incremental MST,
as well as their theoretical bounds enabled by our new algorithm.

The \ourtree{} is also easy to implement.
Our source code is publicly available \cite{amtreecode}.
We tested different versions of \ourtree{} in the scenario of temporal graph processing. 
We compare \ourtree{} against the solution using link-cut tree~\cite{sleator1983data}, and a recent solution using OEC-forest~\cite{song2024querying}.
As discussed, the link-cut tree provides strong theoretical bounds, but may incur high overhead in practice.
OEC-forest was proposed as a more practical solution, but has no theoretical guarantee.
\ourtree{} achieves the same theoretical guarantee as link-cut tree, and also achieves strong performance in practice.
Overall, our lazy \ourtree{} based on \stitchfunc{} gives the best performance---
on average across seven tested graphs, its updates are 8.7$\times$ faster than link-cut tree and 1.2$\times$ faster than OEC-forest, 
and queries are 10.4$\times$ faster than link-cut tree and 2.0$\times$ faster than OEC forest. 
We summarize the contributions of this paper in \cref{fig:outline}.

\begin{figure}
  \centering
  \small
  \includegraphics[width=\columnwidth]{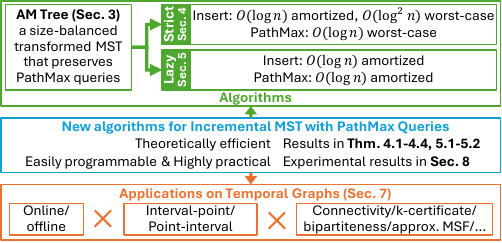}
  \caption{Outline and contributions of this paper.}\label{fig:outline}
\end{figure}


\section{Preliminaries}


\subsection{Graphs and Minimum Spanning Trees} 
Given a graph $G=(V,E)$, we use a triple $(u,v,w)$ to denote an edge in $E$ between $u$ and $v$ with weight $w$. 
With clear context we also use $(u,v)$ and omit the weight $w$. 
We use $n=|V|$ as the number of vertices. 
For simplicity, throughout this paper we assume that the edge weights are \emph{distinct}. 
In practice we can always break ties consistently. 
For a path $P$ in $G$, we use $\max(P)$ to denote the maximum edge weight in $P$. 

Given a weighted undirected graph $G = (V, E)$,
the minimum spanning tree (MST) is a subgraph $T = (V, E')$ such that $E' \subseteq E$
and $T$ is a tree that connects all vertices in $V$ with minimum total edge weight.
More generally, the minimum spanning forest (MSF) problem is to compute an MST for every connected component of the graph. 

In a rooted tree, the \emp{depth} of a node is the number of its ancestors in the tree. 
The \emp{height} of a (sub)tree is the longest hop distance from it to any of its descendants. 
The \emp{size} of a (sub)tree is the number of nodes in the tree. 
We use \emph{node} and \emph{vertex} interchangeably in this paper. 

\vspace{-.1in}
\subsection{Temporal Graph and Path-Max Queries} 
\label{sec:prelim:temporal}
Throughout the section, we will use one specific problem to introduce the connection between temporal graphs and MST, with an illustration given in \cref{fig:example}. 
This problem, which we refer to as the \emp{point-interval temporal connectivity}, considers a
temporal graph $G^*$ where each edge $e$ is associated with a timestamp $t(e)$. 
A query $(u,v,t_1,t_2)$ considers all edges with timestamp in $[t_1,t_2]$ and determines whether $u,v\in V$ are connected by these edges. 
To do this, one can maintain an auxiliary dynamic graph $G$\footnote{Note that the main technique of this paper is to design efficient algorithms for maintaining the MST for the \emph{auxiliary graph}, which is more often referred to in this paper. 
To avoid confusion, we use the notation $G$ to denote the auxiliary graph and use $G^*$ to denote the original temporal graph. } 
such that an edge $e$ is added to $G$ at time $t(e)$ with weight $-t(e)$~\cite{crouch2013dynamic}. 
We use $G_t$ to denote the status of the auxiliary graph at time $t$. 
With clear context we drop $t$ and directly use $G$. 
If a path $P$ in $G$ between two vertices $u$ and $v$ has maximum edge weight $\max(P)=w$, 
it means that all edges on the path are added after time $|w|$. 
To consider all paths between two vertices to determine connectivity, 
we define the \pathmax{} query on a graph $G$ as follows. 

\begin{definition}[Path-Max]
  Given a graph $G=(V,E)$, 
  the path-max query on two vertices $u,v\in V$ is defined as
  $\pathmax_G(u,v)=\min \{\max\{w\mid (x,y,w)\in P\}\}$ where $P$ is any path connecting $u$ and $v$.
  With clear context we drop the subscript $G$ and only use $\pathmax(u,v)$. 
\end{definition}

To determine whether $u,v\in V$ are connected by edges within time $[t_1,t_2]$, 
one can compute $w=\pathmax(u,v)$ on the auxiliary graph $G_{t_2}$, which only contains edges appearing before time $t_2$. 
If $|w|>t_1$, then there exists a path $P$ such that all edges on $P$ appear after $t_1$, and thus $u$ and $v$ are connected. 
Otherwise $u$ and $v$ are disconnected. 
\cref{fig:example} shows an example of how to use incremental MST to solve the point-interval temporal connectivity problem.

\begin{figure*}[h]
  \centering
  \includegraphics[width=1.0\textwidth]{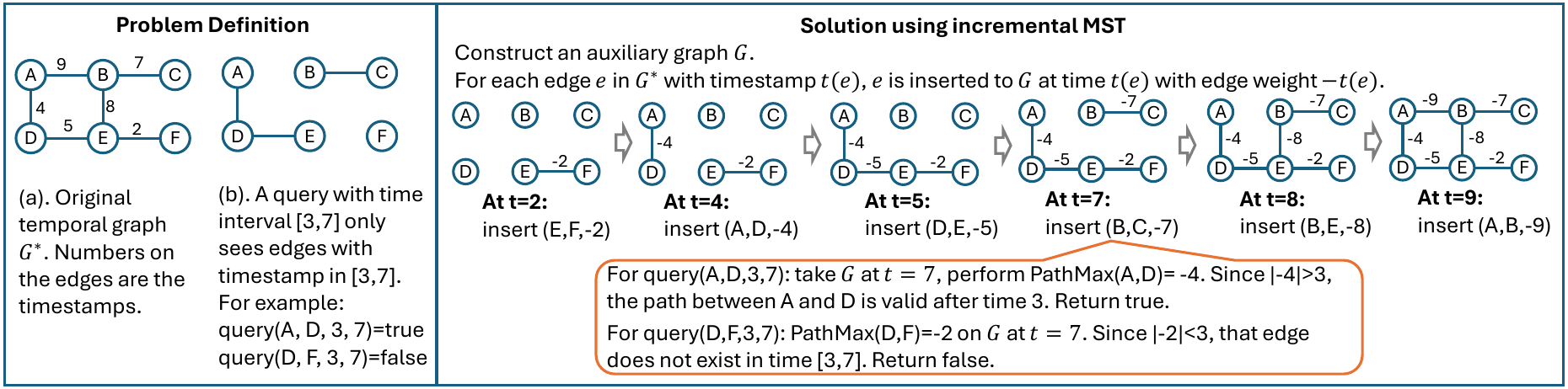}
  \caption{Solving point-interval temporal connectivity on graph $G^*$ using incremental MST \cite{crouch2013dynamic}.}\label{fig:example}
\end{figure*}

To answer path-max queries, one can generate another (usually sparser) graph to accelerate queries. 
We say two graphs $G=(V,E)$ and $G'=(V,E')$ are \emp{path-max equivalent}, or \emp{\pmeq{}}, if $\forall u,v\in V, \pathmax_{G}(u,v)=\pathmax_{G'}(u,v)$. 
Existing work has proved the following fact~\cite{crouch2013dynamic}. 

\begin{fact}[\cite{crouch2013dynamic}]\label{fact:mst}
The MST of a graph $G$ is \pmeq{} to $G$.
\end{fact}

Converting \pathmax{} queries on a graph to its MST simplifies the problem, since only one path exists between any two vertices in the MST. 

\vspace{-.1in}
\subsection{Incremental MST} Given a graph $G=(V,E)$, starting with $n$ vertices and no edges, a data structure is designed to support the following operations:

\begin{itemize}[noitemsep]
  \item \insertfunc$(u,v,w)$: insert an edge $(u,v,w)$ into the graph.
  \item \pathmax$(u,v)$: report the maximum edge weight on the path between $u$ and $v$ on the MST.  
  \item \fname{ReportMST}$()$: report the current MST. Such a query may require to report the total weight or to determine whether an edge is in the MST. 
\end{itemize}

Based on the discussions in Sec. \ref{sec:prelim:temporal} and Fact. \ref{fact:mst}, we can convert the aforementioned point-interval temporal connectivity problem to an incremental MST problem. The main contribution of this paper is to support efficient incremental MST both in theory and in practice, thus leading to improved solutions to temporal graph applications. 

In a graph $G$, the edge with the largest weight on a cycle is not included in the MST (the \emp{red rule}~\cite{Tarjan83}). 
Thus, when inserting edge $(u,v,w)$, many existing incremental MST algorithms~\cite{song2024querying,anderson2020work} find the maximum edge weight between $u$ and $v$ in the current tree, and replace it with the new edge if $w$ is smaller.
Our algorithm also makes use of this idea. 

\hide{
\subsection{Persistent Data Structures}

A persistent data structure \cite{persistence} keeps history versions when being modified. 
In this paper, we also persist \ourtree{} using the standard approach of version lists. We present more details in \cref{sec:persistent}. 
}

\hide{
\subsection{Single-Linkage Dendrogram}

Single-linkage clustering is a hierarchical clustering method
that groups data points based on the minimum distance between clusters \cite{schutze2008introduction}.
It is a bottom-up approach that starts with each data point as a singleton cluster,
and iteratively merges the two closest clusters until a single cluster remains.
The output of single-linkage clustering is called the single-linkage dendrogram (SLD).
The dendrogram is a binary tree where the leaf nodes are the singleton clusters (original data points),
and the internal nodes represent the clusters formed by merging the two closest clusters.
Obviously, this process mimics the Kruskal's algorithm \cite{kruskal1956shortest} for computing the MST.

Single-linkage dendrograms are widely used in various fields to analyze the hierarchical structure of real-world data \cite{}.
}

\section{The \ourtreefull{}}

%

In this section, we propose the \ourtree{} to support incremental MST. 
Recall that an incremental MST needs to maintain the edges in the MST and efficiently answer \pathmax{} queries. 
To make the queries and updates efficient, we want to keep the tree diameter small, specifically $O(\log n)$.
However, this is not easy since the MST itself may have a large diameter---it can even be a chain of length $n-1$. 
Hence, we first introduce the concept of a \tmstfull{} (\tmst{}), and propose our solution, the \ourtreefull{} (\ourtree), based on it.

\begin{definition}[\titlecap{\tmstfull{}} (\tmst)]
\label{def:tmst}
  Given a connected weighted graph $G=(V,E)$ and its minimum spanning tree $\original{T}=(V,\original{E})$.
  A \tmstfull{} (\tmst{}) of $\original{T}$ is a tree $T=(V,E)$ with the following properties:
  \begin{itemize}[noitemsep]
    \item The vertex set in $T$ is the same as $\original{T}$.
    \item There is a one-to-one mapping between $E$ and $\original{E}$, such that the weights of corresponding edges are the same. 
    \item 
    $\forall u,v\in V$, $\pathmax_T(u,v)=\pathmax_{\original{T}}(u,v)$.
  \end{itemize}
\end{definition}


For simplicity, we use the same term \tmst{} to refer to the transformed minimum spanning forest, if the graph is disconnected. 
We say a \tmst{} is \emp{valid} or \emp{correct} if it satisfies the invariants in \cref{def:tmst}. 
We give an example of such a transformation in \cref{fig:tmst}.
Note that, although there is a one-to-one mapping between both the vertices and edges of $T$ and $\original{T}$, 
the corresponding edges may or may not be linking two corresponding vertices. 
For example, in \cref{fig:tmst}, the edge $(b,e,3)$ in the MST corresponds to edge $(a,d,3)$ in the \tmst{}. 

\begin{figure}
  \centering
  \includegraphics[width=0.8\columnwidth]{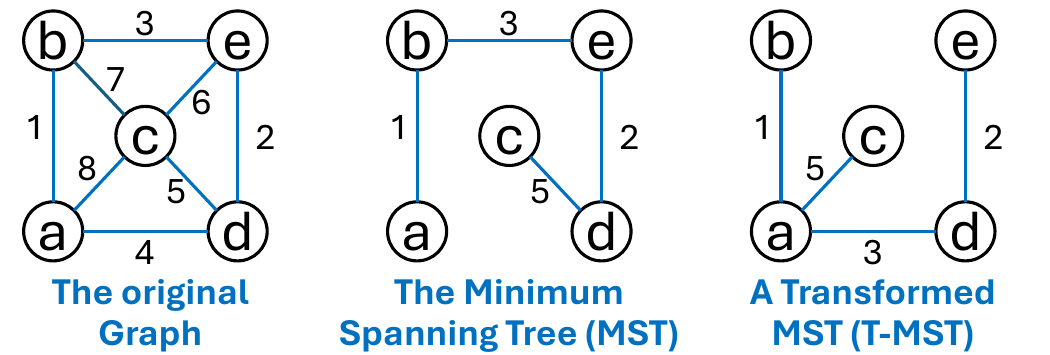}
  \caption{An example of the \tmstfull{} (\tmst). 
  A \tmst{} redistribute the edges in an MST, but preserves the answers to path-max queries in the MST. }\label{fig:tmst}
\end{figure}


\hide{
The goal to transform $\original{T}$ to $T$ is to support efficient \pathmax{} queries. 
To achieve this, it is beneficial to maintain a low tree diameter, such that the \pathmax{} query only checks a small number of edges. 
We say a \tmst{} is \emp{balanced} if its diameter is $O(\log n)$. 
Similarly, organizing the tree as a rooted structure can further facilitate \pathmax{} queries. 

Our proposed data structure, \ourtree{}, is a rooted, balanced \tmst{} structure. It uses a weight-balanced strategy to control the tree height. 
}

The goal of transforming $\original{T}$ to $T$ is to achieve a low diameter, such that a path-max query can simply check all edges on the path.
Similarly, organizing the tree as a rooted structure can facilitate \pathmax{} queries. 
Below, we define \ourtree{}, which is a rooted, size-balanced \tmst{} structure. 
In \ourtree{}, each node $u$ maintains the following information: $\parent{u}$ (the parent of $u$), $\size{u}$ (the subtree size of $u$), and $\weight{u}$ (the edge weight between $u$ and its parent). 

\begin{definition}[\ourtreefull{} (\ourtree{})]
\label{def:ours}
  Given a connected weighted graph $G=(V,E)$, an \ourtree{} is a rooted \tmst{} such that for each (non-root) node $u$, 
  \begin{equation*}\label{eq:weight-balanced}
    \size{u}\le (2/3)\cdot\size{\parent{u}}~~ \text{(Anti-Monopoly Rule)}
  \end{equation*} 
\end{definition}

The key property of the \ourtree{} is the anti-monopoly rule, which disallows any child's size to be a factor of 2/3 or larger than its parent. This guarantees $O(\log n)$ height of the tree. 

\begin{fact}\label{fact:tree-height}
  In a tree $T$ with size $n$, if all nodes satisfy the anti-monopoly rule, then the height of $T$ is $O(\log n)$. 
\end{fact}

For a node $x$ and its parent $y$, we say $x$ is a \emp{heavy child} of $y$ if $\size{x}>(2/3)\size{y}$. 
A node $y$ is \emp{unbalanced} if it has a heavy child, and is \emp{balanced} or \emp{size-balanced} otherwise. 



\subsection*{The \promotefunc{} primitive for the \ourtree{}} 
To ensure the anti-monopoly rule, we may need to transform the tree while preserving the \pathmax{} queries. 
We start by showing the \emph{TW transformation} mentioned in~\cite{song2024querying}.

\begin{fact}[TW Transformation~\cite{song2024querying}]\label{lemma:transform}
  Given a graph $G=(V,E)$ and two edges $(x,y,w_1)$ and $(y,z,w_2)$ in $E$ such that $w_1\ge w_2$,
  the \pathmax{} queries on $G$ are preserved if we replace the edge $(x,y,w_1)$ with edge $(x,z,w_1)$. 
\end{fact}

Note that this is also simply true on a T-MST.
Based on this observation, we define a \emph{\promote} operation on the \ourtree{}.
$\promotefunc(x)$ promotes node $x$ one level up (closer to the root) without affecting the \pathmax{} queries of the tree. 
We illustrate this process in \cref{fig:rotate}. 
Let $y$ be the parent of $x$, and $z$ the parent of $y$. 
\promotefunc{}$(x)$ executes one of the two following operations to promote $x$, both of which are TW-transformations. 

\begin{itemize}
  \item \textbf{\titlecap{\shortcut{}}}. If $w_1>w_2$, $x$ is directly \promoted{} to be $z$'s child, still with edge weight $w_1$. $y$ now becomes a sibling of $x$. 
  \item \textbf{\titlecap{\zigzag{}}}. If $w_1< w_2$, or if $y$ is the root, $y$ is pushed down to be $x$'s child, still with edge weight $w_1$. If $y$ is not the root, $x$ will be attached to $z$ as a child with edge weight $w_2$.
\end{itemize}


\begin{figure}[t]
  \centering
  \includegraphics[width=0.8\columnwidth]{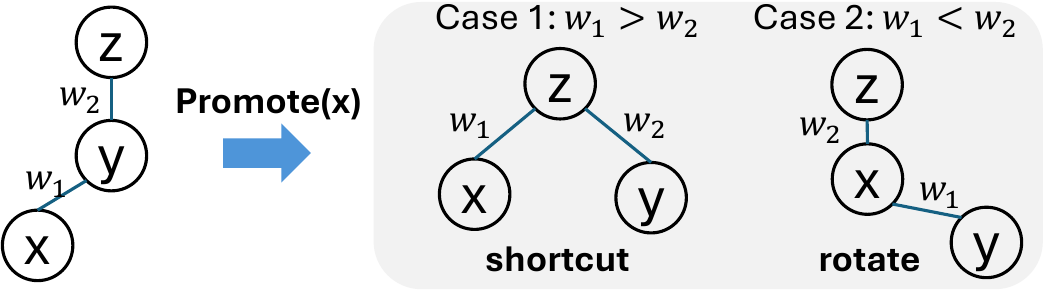}
  \caption{An illustration of the \promotefunc{} algorithm.
  \label{fig:rotate}}
\end{figure}

\hide{
\begin{algorithm}[h]
\small
\DontPrintSemicolon
\caption{The \promotefunc{} function for \ourtree{}.\label{algo:rotate}}
\SetKwProg{myfunc}{Function}{}{}
\SetKwFor{parForEach}{ParallelForEach}{do}{endfor}
\SetKwFor{mystruct}{Struct}{}{}
\SetKwFor{pardo}{In Parallel:}{}{}
\SetKwInput{Input}{Input}
\Input{A node $x$}

\yihan{we can probably drop this pseudocode. I think the figure is clearer than the code. }

Let $(x,y,w_1)$ be the edge connecting $x$ and its parent $y$, and $(y,z,w_2)$ the edge connecting $y$ and its parent $z$. \\
\If(\tcp*[f]{shortcut}){$y$ is not the root and $w_1\ge w_2$}{
  $\parent{x}\gets z$\\
  $\size{y}\gets \size{y}-\size{x}$\\
}
\Else(\tcp*[f]{rotate}) {
  $\parent{x}\gets z$\\
  $\parent{y}\gets x$\\
  swap$(\weight{x}, \weight{y})$\\
  $\size{y}\gets \size{y}-\size{x}$\\
  $\size{x}\gets \size{x}+\size{y}$\\
}
  
\end{algorithm}
}

The \promotefunc{} operation is an important building block to both the correctness and the efficiency of \ourtree{}. 
In the next sections, we will discuss efficient algorithms for \ourtree{s}.
We first show a \emp{strict} version of \ourtree{} in Sec. \ref{sec:preemptive}, which always keeps the tree height in $O(\log n)$. 
However, the strict version requires maintaining the child pointers for all nodes, which brings up performance overhead in practice. 
To tackle this, in \cref{sec:lazy} we discuss the lazy version of the \ourtree{}, which only requires keeping the parent pointer of each node. 
By avoiding maintaining child pointers, the lazy version is much simpler, more practical, and easier to program. 
\section{The Strict \ourtree{}}\label{sec:preemptive}

In this section, we present the strict \ourtree{}, where all tree nodes strictly follow the AM-rule at all times.
Recall that an \ourtree{} $T$ supports the following operations: 
\insertfunc{}$(u,v,w)$, which updates the tree to reflect an edge insertion $(u,v,w)$ to the graph,
\pathmax{}$(u,v)$, which gives the maximum edge weight between $u$ and $v$ on the MST,
and \fname{ReportMST}, which reports information of the current MST.


Among them, we only need to design the \insertfunc{}$(u,v,w)$ function that maintains the tree invariants, since \pathmax{} and \fname{ReportMST} are read-only. 
We show two solutions to approach this.
The first solution is based on a helper function \perchfunc{}, and is algorithmically simpler.
At a high level, it uses the \perchfunc{} function to \promote{} both $u$ and $v$ to the root, and then connects $u$ and $v$ with weight $w$, if $w$ is smaller than the current edge between $u$ and $v$. 
The second approach is based on \emph{stitching} the paths from $u$ and $v$ to the root without affecting the \pathmax{} results, which is slightly more complicated but practically faster. 
Both algorithms achieve the same theoretical guarantees. 
In \cref{sec:lazy}, we will extend both of them to lazy versions. 


\subsection{The High-Level Algorithmic Framework}

We start with the high-level framework of \ourtree{}, presented in Alg.~\ref{algo:preemptive-cdt}. We will analyze the algorithms in Sec. \ref{sec:preemptive:analysis:correctness} and \ref{sec:preemptive:analysis:efficiency}. 

\myparagraph{Edge Insertion.} The strict \ourtree{} rebalances the tree immediately once it is updated. 
To insert an edge $(u,v,w)$ into an \ourtree{} $T$, the algorithm starts with a function \linkfunc{}$(u,v,w)$, which applies the edge insertion $(u,v,w)$ such that the tree remains valid, but may be unbalanced.
This operation may insert the new edge into $T$, or cause an existing edge on $T$ to be replaced by the new edge, or have no effect to the tree if the new edge $(u,v,w)$ does not appear in the MST of the graph.
The resulting tree is not unique---one can use multiple ways to apply \linkfunc{}. 
We present two \linkfunc{} algorithms: the first one (Sec. \ref{sec:preemptive:perch}) is based on a primitive \perchfunc{}, which is conceptually simpler; the other one (Sec. \ref{sec:preemptive:stitch}) is based on a primitive \stitchfunc{}, which is more complicated but more efficient in practice. 
We prove the correctness of the algorithm in Thm. \ref{thm:exact:correctness}.

The structural changes in the \linkfunc{} operation may cause the tree unbalanced.
We say a node $y$ is \emp{affected} (or may become unbalanced) during the \linkfunc{} operation
if either $y$'s children list is changed,
or the subtree size of any $y$'s child is changed.
We will show that all such nodes are on the path from $u$ or $v$ to the root before the \linkfunc{} operation.
We collect all these nodes in a set $S$.
Next, a \downwardcalifunc{} function is applied on each node $y$ in $S$.
\downwardcalifunc{}$(y)$ aims to ensure that node $y$ achieves a balance with all its children. This operation first identifies whether $y$ has a heavy child $x$. If so, $x$ will be promoted and removed from $y$'s subtree. This process is repeated until $y$ is balanced. 
In Thm. \ref{thm:exact:balance}, we prove that the tree becomes balanced after the \insertfunc{} operation.

We note that, to perform \downwardcalifunc{}, we need to store the child pointers in each node, and efficiently determine whether the anti-monopoly rule is violated. 
To help the reader understand the high-level idea more easily, we assume a black box that can determine whether there is a heavy child of a tree node $u$ (and find it if one exists) with $O(1)$ time.
Throughout the description and analysis, we assume the existence of this black box, and we give a possible implementation in \cref{sec:heavy-child}.


\myparagraph{Path-max Queries.} A \pathmax{} query finds the maximum edge on the path between $u$ and $v$ on $T$. 
Relevant edges can be identified by first finding Least Common Ancestor (LCA) of $u$ and $v$ as $l$, and finding all edges from $u$ and $v$ to $l$. 

\myparagraph{Other Queries.} Other MST-related information can be easily maintained during updates. 
For example, we can modify the insertion function to maintain the membership of each edge in the MST.
We can use a boolean flag for each edge to denote if it is in $T$. 
Note that an insertion can only cause one edge to alter in the MST. 
In \linkfunc{}, when inserting an edge $e$ incurs a replacement of another edge $e'$, 
we can directly change the flag of both edges in $O(1)$ extra cost. 
Similarly, one can update the total weight of the MST after each insertion in $O(1)$ cost.

\subsection{\titlecap{The \perch{}}-based Solution}\label{sec:preemptive:perch}

We now present the first implementation of the \linkfunc{} algorithm using the helper function \perchfunc{}. 
We call this algorithm \perchandlink{} and present the pseudocode on Lines \ref{line:perch:start} to \ref{line:perch:end} in Alg. \ref{algo:preemptive-cdt}. 

To insert an edge $(u,v,w)$ into the graph, we may need to update the \ourtree{} $T$ accordingly such that it is still a valid \tmstfull{}. 
Based on the properties of MST, if $u$ and $v$ were not connected before the insertion, the new edge $(u,v,w)$ should just appear in the MST. 
Otherwise, if $u$ and $v$ were previously connected, the MST may be changed due to the new edge insertion. 
In particular, adding edge $(u,v,w)$ may introduce a cycle on the graph, and the largest edge on the cycle should be removed.
The \ourtree{} needs to be updated to reflect such a change in the true MST. 

The \perchandlink{}$(u,v,w)$ algorithm starts by calling a helper function, \perchfunc{}, on both $u$ and $v$. 
The goal of \perchfunc$(x)$ is to restructure the tree and put node $x$ at the top. 
It simply applies \promotefunc{} on $x$, until $x$ becomes the root of the tree. Based on \cref{lemma:transform}, the resulting tree is still a valid \tmst{}, but the tree height may be affected. 
After calling \perchfunc{} on both $u$ and $v$, if $u$ and $v$ were originally disconnected, 
$u$ and $v$ will be made the root of their own tree in the spanning forest. 
Hence, we directly attach $u$ as $v$'s child with the new edge weight $w$. 

If $u$ and $v$ were already connected before the edge insertion, the first \perchfunc{} on $u$ will reroot the tree at $u$,
and the second \perchfunc{} on $v$ will further put $v$ on the top, pushing $u$ down as the child of $v$. 
In this case, we simply check the current edge weight between $u$ and $v$ (stored in $\weight{u}$), and update it to $w$ if $w$ provides a lower value. 

Intuitively, the two \perchfunc{} operations preserve the validity of the \tmst{} before the edge insertion, and then the new edge is directly reflected on $T$ by connecting $u$ and $v$ by weight $w$. 
If $u$ and $v$ were connected before, after \perch{ing} both $u$ and $v$, $u$ and $v$ should be connected by another edge $(u,v,w')$. 
Note that the design of \promotefunc{} preserves the path-max queries.
Hence, since the edge $(u,v,w')$ is the only edge from $u$ and $v$ on $T$, $w'$ is the path-max.
Therefore, if $w<w'$, we replace the old edge with the new edge with weight $w<w'$. 


\begin{algorithm}[t]
\small
\DontPrintSemicolon
\caption{The Strict \ourtree{}\label{algo:preemptive-cdt}}
\SetKwProg{myfunc}{Function}{}{}
\SetKwFor{parForEach}{ParallelForEach}{do}{endfor}
\SetKwFor{mystruct}{Struct}{}{}
\SetKwFor{pardo}{In Parallel:}{}{}

\tcp{We omit the maintenance of the $\size{\cdot}$ array for simplicity}

\myfunc(){\upshape\insertfunc($u,v,w$)}{
  $S\gets\{u,v\}\cup\{\text{all ancestors of }u\}\cup\{\text{all ancestors of }v\}$\\
  \linkfunc$(u,v,w)$\tcp*[f]{Plug in \perchandlink{} or \stitchandlink{}}\\
  \lForEach{node $y\in S$}{
    \downwardmaintain$(y)$
  }
}

\myfunc{\upshape\pathmax$(u,v)$}{
  return the maximum edge weight on the path from $u$ to $v$\\
}

\medskip

\tcp{\perchfunc{}-based \linkfunc{} function}

\myfunc{\upshape\perchandlink$(u,v,w)$\label{line:perch:start}}{
  \perchfunc$(u)$\\
  \perchfunc$(v)$\label{line:exact:finish-perch}\\
  \lIf{$\parent{u}=v$\label{line:exact:if-condition}} {
    $\weight{u}\gets \min(\weight{u},w)$
  }
  \Else(\tcp*[f]{$u$ and $v$ were previously disconnected}){
    $\parent{u}\gets v$\\
    $\weight{u}\gets w$\\
  }
}

\myfunc{\upshape$\perchfunc(x)$}{
  \lWhile{$\parent{x}\ne$ null}{
    $\promotefunc(x)$\label{line:perch:end}
  }
}

\medskip

\tcp{\stitchfunc{}-based \linkfunc{} function}

\myfunc{\upshape\stitchandlink$(u,v,w)$\label{line:stitch:start}}{
  \lIf{$u=v$ \textbf{or} $u=$ null \textbf{or} $v=$ null}{\Return}
  \ElseIf{$\parent{u}\ne$ null \textbf{and} $w>\weight{u}$}{\stitchandlink{}$(\parent{u},v,w)$}
  \ElseIf{$\parent{v}\ne$ null \textbf{and} $w>\weight{v}$}{\stitchandlink{}$(u,\parent{v},w)$}
  \Else {
    \tcp{Use the fixed subtree size before the insertion}
    \lIf{$\size{u}>\size{v}$}{swap$(u,v)$\label{line:stitch:stop}}
    $u'\gets \parent{u}$\\
    $w'\gets \weight{u}$\\
    $\parent{u}\gets v$\label{line:stitch:change1}\\
    $\weight{u}\gets w$\label{line:stitch:change2}\\
    \stitchandlink$(u',v,w')$\label{line:stitch:end}\\
  }
}

\medskip
\myfunc{\upshape{\downwardmaintain}($y$)}{
  \While{\upshape $y$ has a child $x$ such that $\size{x}>\frac{2}{3}\size{y}$}{
    $\promotefunc(x)$\\
  }
}

\end{algorithm}


\subsection{\titlecap{\stitch}-based Solution}\label{sec:preemptive:stitch}

Our second approach for insertion is based on the idea of stitching the tree paths from both $u$ and $v$ to the root. 
The pseudocode is presented on Lines \ref{line:stitch:start} to \ref{line:stitch:end}, and an illustration is shown in \ref{fig:stitch}. 
Instead of relying on \promotefunc{}, this approach directly adds the edge $(u,v,w)$ (for an insertion) to the tree, and uses TW transformations to move this edge to its final destination and accordingly restructure the tree.
This approach is slightly less intuitive, but performs faster in practice since it can touch fewer vertices in this process.

Based on TW transformation, if $\weight{u}<w$, we can replace the edge with $(\parent{u},v,w)$, and recursively call \stitchandlink$(\parent{u},v,w)$. 
We do the same thing for $v$. 
When this process ends (Line \ref{line:stitch:stop}), the recursive call must have reached two vertices $u$ and $v$ such that $\weight{u}>w$ and $\weight{v}>w$.
To connect $u$ and $v$ with edge weight $w$, we attach the one with smaller subtree size as a child to the larger one. 
Later in the analysis, we will show that this is important to bound the amortized cost of this algorithm. 
WLOG we assume $\size{u}<\size{v}$ (swap them otherwise). In this case, we will reassign the parent of $u$ to be $v$ with edge weight~$w$.

By doing this, on the current tree $T$, $u$ is connected to both its original parent $u'$ and its new parent $v$. 
Let the weight of the original edge between $u$ and $u'$ be $w'$. 
Then this edge $(u,u',w')$ is not a valid tree edge anymore, and we need to relocate it in the tree. 
Consider the two edges $(u,u',w')$ and $(u,v,w)$.
Based on TW transformation, we can equivalently move $(u,u',w')$ to $(u',v,w')$ since $w'>w$.
Hence, the algorithm finally calls \stitchandlink$(u',v,w')$ to finish the process.

Finally, there are two base cases. First, when $u$ and $v$ were connected before, 
then by moving edges up, $u$ and $v$ in the recursive calls will finally move to their LCA in the tree and become the same node. 
In that case, we do not need to further connect them and can terminate. 
The second case is when they were not in the same tree. Then by the recursive calls, one of them will reach the root, and the parent in the recursive call becomes null. In that case, the algorithm can also terminate, since one of the trees has been fully attached to the other.

\begin{figure}
  \centering
  \includegraphics[width=\columnwidth]{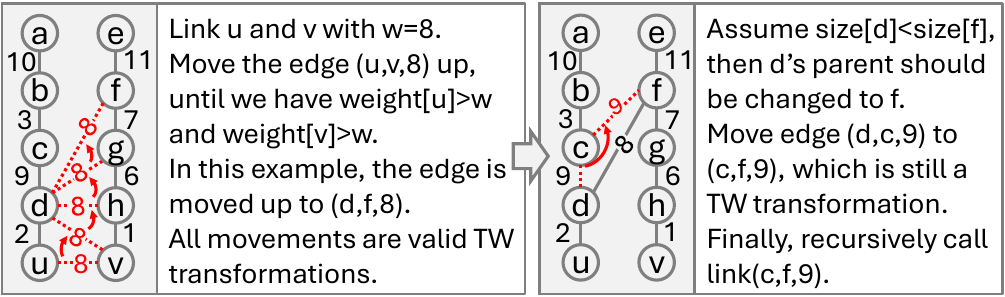}
  \caption{An example of \stitchandlink{}. The figure illustrates \stitchandlink$(u,v,8)$. Values on the edges are edge weights. 
  The figure shows all vertices on the path from $u$ and $v$ to the root, and omits all other vertices. An explanation about the process is shown in the figure, and the pseudocode is presented in Alg. \ref{algo:preemptive-cdt} Lines \ref{line:stitch:start}-\ref{line:stitch:end}. 
  }\label{fig:stitch}
\end{figure}


\subsection{Correctness Analysis}\label{sec:preemptive:analysis:correctness}

We first show the correctness of the algorithm, i.e., after the insertion algorithm, the tree 1) is a valid \tmst{} that handles the insertion of edge $(u,v,w)$ to the original graph, and 2) satisfies the AM rule. 

\vspace{-.1in}

\begin{theorem}[Correctness]
\label{thm:exact:correctness}
  Given a graph $G=(V,E)$ and a \tmst{} $T=(V,E_T)$ for $G$, after \insertfunc{}$(u,v,w)$ in Alg. \ref{algo:preemptive-cdt}, using either \perchandlink{} or \stitchandlink{}, $T$ is a valid \tmst{} for $G'=(V, E\cup \{(u,v,w)\})$. 
\end{theorem}
\begin{proof}
  To show correctness, we need to verify that the results of path-max queries on any two vertices are still preserved. Note that all modifications in \downwardcalifunc{} only use \promotefunc{}, which are all TW transformations and preserve the path-max results.
  Therefore, we only need to show that both \perchandlink{} and \stitchandlink{} preserve path-max results.
  
  Consider we directly add edge $(u,v,w)$ on $T$, and get a graph $T'=(V,E_T\cup \{(u,v,w)\})$. 
  Consider the \linkfunc{} algorithm applied to $T$, and we perform the same operations on $T'$. 
  Then all path-max queries on $T'$ should be the same as that on $G'$, so $T'$ is \pmeq{} to $G$.
  We will show that the final result $T$ is \pmeq{} to $T'$. 
  Note that $T'$ may or may not be a tree, depending on whether $u$ and $v$ were connected before. 
  
  For \perchandlink{}, $T'$ is exactly the tree obtained after Line \ref{line:exact:finish-perch} augmented with an additional edge $(u,v,w)$. 
  The final step is the if-conditions from Line \ref{line:exact:if-condition}. 
  In the first case where $u$ and $v$ were not connected before, 
  $u$ is set as the child of $v$ with weight $w$, obtaining a tree $T$ that is the same as $T'$.
  In the second case, $u$ and $v$ were in the same tree. Therefore after the two \perchfunc{} operations, 
  $u$ should be connected with $v$ with an existing edge $(u,v,w')$, and $T'$ further augments an edge $(u,v,w)$ to $T$. 
  In this case, only the lower weight should be kept in the MST, and therefore the algorithm selects the minimum of the original weight $w'$ and the new weight $w$. 
  
  For \stitchandlink{}, the algorithm exactly first augments $T$ with the virtual edge and gets $T'$. All later edge movements are TW transformations, as discussed in Sec. \ref{sec:preemptive:stitch}. Therefore, during \stitchandlink{}, $T$ is always \pmeq{} to $T'$. 
  The only exception is the base case where $u=v$, that the edge with weight $w$ will be dropped in $T$. 
  In this case, conceptually this edge in $T'$ is a self-loop on node $u=v$. Therefore, omitting it does not change results for path-max queries. 
\end{proof}

\vspace{-.1in}

We then present the theorem below, which states that the tree stays balanced after insertion. 
\ifconference{Due to page limit, we defer the proof to the full version of this paper \cite{amtreefullversion}.}
The key proof idea is to verify that \downwardcalifunc{}$(b)$ will always
fix the imbalance at node $b$, without introducing other unbalanced nodes. 
Therefore, calling \downwardcalifunc{} on
all affected nodes in the previous process guarantees to rebalance the tree. 

\begin{theorem}[Balance Guarantee]
\label{thm:exact:balance}
  After each \insertfunc{} operation, all nodes in the \ourtree{} are balanced, and the \ourtree{} has $O(\log n)$ height.
\end{theorem}

\iffullversion{\begin{proof}
  We first show that in each \insertfunc{} operation,
  the set $S$ contains all affected nodes during the \perchandlink{} operations.
  Recall that a node $y$ is affected (or may become unbalanced)
  either when $y$'s children list is changed,
  or when the subtree size of any $y$'s child is changed.
  For \perchandlink{}$(u, v)$, we perform \perchfunc{} on $u$ and $v$,
  where we call a series of \promotefunc{} operations until the node becomes the root.
  In each call to \promotefunc{}($x$), only $x$, $y$ or $z$ can be affected (see \cref{fig:rotate}),
  which are all $x$'s ancestors,
  and no other nodes become $x$'s ancestor.
  So \perchfunc{}$(u)$ will only affect $u$ and all its ancestors.
  After \perchfunc{}$(u)$, the new root $u$ may become a new ancestor of $v$.
  Combining the two \perchfunc{} operations, only $u$, $v$ and all their ancestors in the original tree can be affected.
  For \stitchandlink{}$(u, v)$,
  note that during each recursive call,
  only the subtrees of $u'$ (removing child $u$) and $v$ (obtaining child $u$) are changed,
  so $u'$, $v$ and all their ancestors are affected.
  In each recursive call, we move either $u$ or $v$ to a higher level,
  so all such affected nodes are on the path from $u$ or $v$ to the root before the \stitchandlink{} operation.

  Now we prove Thm. \ref{thm:exact:balance} inductively that after each insertion,
  all nodes are still size-balanced. 
  At the beginning when there is no edge, the conclusion trivially holds.
 
  Assume the \insertfunc{} function starts on a tree where all nodes are size-balanced. 
  Note that during the algorithm, only the nodes in $S$ may be affected and become unbalanced due to the \perchfunc{} operations. 
  At the end of the algorithm, we perform a \downwardcalifunc{} operation on all nodes in $S$. 
  This function repeatedly fixes the imbalance issue on each node until it does not have a heavy child.
  
  The key point of the proof is that the \downwardcalifunc{}$(y)$ function will always make node $y$ balanced, without introducing more unbalanced nodes. 
  In \downwardcalifunc{}, if we find a heavy child $x$ of $y$, we will call \promotefunc{} on $x$. 
  Based on the \promotefunc{} algorithm (see \cref{fig:rotate}), 
  when we \promote{} $x$, the node $x$ itself and $y$'s parent $z$ may become affected during this operation.
  However, if we assume $x$ and $z$ are balanced before the promotion of $x$,
  we can show that $x$ and $z$ will stay balanced after the promotion of $x$ in both \shortcut{} and \zigzag{} cases.
  
  In the \shortcut{} case, $x$'s entire subtree does not change, and therefore $x$ will remain balanced after \promotefunc{}.
  Let $s(\cdot)$ denote the size of a subtree before \shortcut{}, and $s'(\cdot)$ the size of a subtree after \shortcut{}. 
  Since the original tree is balanced, $s(x)<s(y)\le(2/3)s(z)$. 
  In the new tree, since the subtree at $a$ remains unchanged, we have $s'(x)=s(x)<s(y)\le(2/3)s(z)$.
  Since $x$ has been separated out from $y$, $s'(y)=s(y)-s(x)<s(y)\le(2/3)s(z)$.
  Namely, after shortcut, neither $x$ nor $y$ is a heavy child of $z$. For all other children of $z$, their ratio to $z$ stays unchanged.
  Therefore, both $x$ and $z$ remain balanced after a \shortcut{}. 
  
  In \zigzag{}, note that the subtree sizes for $z$ all remain unchanged, so $z$ trivially remains balanced.
  The operation rotate puts $y$ (along with all its subtrees other than $x$) as a subtree of $x$. 
  Note that here we call \promotefunc{} in a \downwardmaintain{} because $x$ was a heavy child of $y$ in the original tree, meaning $s(x)>(2/3)s(y)$. 
  In the new tree, $s'(y)=s(y)-s(x)$, and $s'(x)=s(y)$. 
  Therefore, $s'(y)=s(y)-s(x)<(1/3)s(y)=(1/3)s'(x)$, which means that $y$ is not a heavy child of $x$ in the new tree.
  For each of the other children of $x$, its ratio can only decrease since $y$ has been added to $x$. In summary, all $x$'s children remain valid and $x$ is still balanced.
  
  Note that after the promotion of node $x$, node $y$ may still have another heavy child.
  In this case, \downwardcalifunc{} will repeatedly work on $y$ to find all heavy children and \promote{} them. 
  
  So far, we have proved that each \downwardcalifunc{}$(y)$ only eliminate the possible imbalance at $y$ without introducing other unbalanced nodes. 
  Therefore, applying \downwardcalifunc{} on all nodes in $S$ one by one will finally rebalance all nodes in $T$. 
  If all nodes are size-balanced, from \cref{fact:tree-height}, we know the tree has height $O(\log n)$. 
\end{proof}
}  

\subsection{Cost Analysis}\label{sec:preemptive:analysis:efficiency}

We now prove the cost bounds for the strict \ourtree{}.
Let $d(\cdot)$ be the depth of a node.
We first show the worst-case cost of the two \linkfunc{} functions.

\begin{lemma}\label{lem:link-cost}
  The worst-case cost of \perchandlink{} or \stitchandlink{} is $O(d(u)+d(v))$.
\end{lemma}
\begin{proof}
  The simpler case is \stitchandlink{}. In each recursive call, the algorithm reassigns $u$ or $v$ to another node on a higher level.
  So the worst-case cost is trivially $O(d(u)+d(v))$.

  For \perchandlink{}, we first show that after \perchfunc{}$(u)$, the depth of any node can increase by at most 1.
  \perchfunc{}$(u)$ performs a series of \promotefunc{} operations on $u$.
  In a \promotefunc{} call, let $y$ be the parent of $u$.
  Only the nodes in $y$'s subtree may have their depth increased by 1 (the \zigzag{} case in \cref{fig:rotate}).
  After that, $u$ is promoted one level up, and $y$ can never be the parent of $u$ again.
  Therefore, the depth of any node can increase by at most 1.

  In \perchandlink{}, we \perch{} both $u$ and $v$ and then connect $u$ and $v$.
  The latter part takes constant time, so we only need to consider the cost of \perch{ing} $u$ and $v$.
  The function \perchfunc{}$(u)$ performs $d(u)$ calls to \promotefunc{}($u$),
  each of which decreases $d(u)$ by 1.
  So \perchfunc{}$(u)$ takes $O(d(u))$ time.
  After \perch{}ing $u$, the depth of $v$ is increased by at most 1.
  Therefore, \perchfunc{}$(v)$ takes $O(d(v))$ time.
  Combining all the above, the worst-case cost for \perchandlink{} is $O(d(u)+d(v))$.
\end{proof}

We then show the worst-case cost of the \insertfunc{} operation.

\begin{theorem}\label{thm:exact:worst-case}
  The worst-case cost for \insertfunc{} is $O(\log^2 n)$. 
\end{theorem}

\begin{proof}
  The total cost for \insertfunc{} includes the cost for \linkfunc{} and \downwardcalifunc{}.
  Based on \cref{lem:link-cost}, the cost for \linkfunc{} is $O(d(u)+d(v))=O(\log n)$.
  This also means that the number of affected nodes in $S$ is also $O(\log n)$.
  
  To \calibrate{} each node $y\in S$, we use \promotefunc{} to remove a heavy child from $y$.
  This means that the size of $y$ is reduced by at least a factor of $2/3$.
  Thus, at most $O(\log n)$ \promotefunc{} functions are used in \downwardcalifunc{}($y$).
  Therefore, the worst-case cost for the \insertfunc{} is $O(\log^2 n)$. 
\end{proof}

Next, we use amortized analysis to show that the \insertfunc{} and \pathmax{} operations take $O(\log n)$ amortized time. 
We define the potential function for a node $u$ as:
\begin{equation}
  \phi(u)=\log \size{u}
\end{equation}
We also define the potential function for the whole tree as:
\begin{equation}
\label{eq:potential}
  \Phi(T)=\sum_{i=1}^n\phi(i)=\sum_{i=1}^n\log\size{i}
\end{equation}
The potential function is always non-negative, 
and the potential of the whole tree is $O(n\log n)$. 
Recall that the amortized cost for an operation $\op$ is $\amortized{\op}=\actual{\op}+\Delta(\Phi(T))$,
where $\actual{\op}$ is the actual cost (number of instructions) in the operation $\op$,
and $\Delta(\Phi(T))$ is the change of potential function after the operation. 
We first prove the following important lemma, which states that, 
if a promotion is performed due to imbalance, the amortized cost of \promotefunc{} is free.
In other words, the cost of the \promotefunc{} can be fully charged to previous operations that increase the potential of the tree. 

\begin{lemma}\label{lem:unbalancedrotate}
  If $\size{x}>(2/3)\cdot\size{\parent{x}}$, the operation \promotefunc$(x)$ has zero amortized cost. 
\end{lemma}

\begin{proof}
  We first show that in both \zigzag{} and \shortcut{},
  the potential of the tree will decrease by at least 1.
  Let $y=\parent{x}$.
  Note that during \promotefunc{}($x$), only the potential for $x$ and $y$ will change. 
  Let $s(\cdot)$ be the size of a node before \promotefunc{}, and $s'(\cdot)$ the size after. 
  Based on the assumption in the lemma, $s(x)>(2/3)s(y)$. 
  In a \shortcut{} case, $x$'s potential remains unchanged, and the size of $y$ decreases by at least a factor of $2/3$, causing its potential to decrease by $\log_2 3>1$. 
  In a \zigzag{} case, $s'(x)=s(y)$. 
  For $y$, we have $s'(y)=s(y)-s(x)<(1/2)s(x)$. 
  The potential change after a \promotefunc{} is $(\log s'(x) + \log s'(y)) - (\log s(x)+\log s(y))<\log s(y) + \log (1/2)s(x) - \log s(x) - \log s(y)=-1$. 
  Combining the actual cost and the potential change, \promotefunc{}$(x)$ has zero amortized cost when $\size{x}>(2/3)\size{\parent{x}}$.   

  Here we are assuming the cost for \promotefunc{}$(x)$ is 1. 
  More precisely, suppose the actual cost of \promotefunc$(x)$ is a constant $c$.
  If we use potential function $\phi'(x)=c\cdot\log \size{x}$,
  we will have $\amortized{\op}=\actual{\op}+\Delta(\Phi(T))=c+(-c)=0$.
\end{proof}

From \cref{lem:unbalancedrotate}, we have the following conclusion. 

\begin{lemma}\label{lem:amortized-downcali}
  Assume identifying the heavy child of a node has $O(1)$ cost. Then \downwardmaintain{} has $O(1)$ amortized cost. 
\end{lemma}

\begin{proof}
  The \downwardmaintain{} operation is a sequence of \promotefunc{} operations on a node $x$ such that $\size{x}>(2/3)\size{\parent{x}}$.
  Based on \cref{lem:unbalancedrotate}, all \promotefunc{} operations have zero amortized cost. 
  Hence, the amortized cost for \downwardmaintain{} is $O(1)$. 
\end{proof}

We now show the amortized cost of the \perchandlink{} and \stitchandlink, which will be used to prove Thm.~\ref{thm:exact}.

\begin{lemma}
\label{lem:amortized-perch-link}
  The \perchandlink{}$(u,v,w)$ operation has $O(d(u)+d(v)+\log n)$ amortized cost.
\end{lemma}
\begin{proof}
  By \cref{lem:link-cost}, the actual cost of \perchandlink{} is $O(d(u)+d(v))$.
  We then prove that the increment of the potential is $O(\log n)$.
  The \perchandlink{} function has three steps: two \perchfunc{} functions and the final step to link $u$ and $v$.
  In the two \perchfunc{} calls, we repeatedly promote $u$ or $v$ to a higher level.
  For both shortcut and rotate cases, the sizes of all other nodes are non-increasing,
  so only $\phi(u)$ and $\phi(v)$ may increase.
  The last step connects $u$ and $v$.
  The only case that may change the potential is when a new edge is established, $u$ becomes a child of $v$,
  and only $\phi(v)$ may increase.
  Combining all steps, the only increment on the potential function is $\phi(u)$ and $\phi(v)$,
  so the increment of the potential function is at most $O(\log n)$.
\end{proof}

\begin{lemma}
\label{lem:amortized-stitch}
  The \stitchandlink{}$(u,v,w)$ operation has $O(d(u)+d(v))$ amortized cost.
\end{lemma}
\begin{proof}
  By \cref{lem:link-cost}, the actual cost of \stitchandlink{} is $O(d(u)+d(v))$.
  Note that here we use the fixed subtree size before the insertion.
  It can be shown that after \stitchandlink{}, the subtree size of any node $x$ is at most doubled,
  and $\phi(x)$ is increased by at most 1.
  Since there are at most $O(d(u)+d(v))$ nodes affected in the algorithm, the potential change is also $O(d(u)+d(v))$.
\end{proof}

In a size-balanced tree, $d(u)$ and $d(v)$ are $O(\log n)$. 
Combining Fact \ref{fact:tree-height}, Lemmas \ref{lem:amortized-downcali}, \ref{lem:amortized-perch-link}, and \ref{lem:amortized-stitch}, 
we have the following theorem for the entire \insertfunc{} function.

\begin{theorem}\label{thm:amortized-insert}
  The amortized cost for each \insertfunc{} is $O(\log n)$ using either \perchandlink{} or \stitchandlink{}. 
\end{theorem} 

The bound of the \pathmax{}$(u,v)$ query is trivially $O(\log n)$ since the tree height is $O(\log n)$. 
We can find the lowest common ancestor (LCA) and compare all edges on the path between $u$ and $v$. 
To summarize, we have the following theorem on the cost bounds for the strict \ourtree{}. 

\begin{theorem}\label{thm:exact}
  The strict \ourtree{} supports \pathmax{} in $O(\log n)$ worst-case cost, and \insertfunc{} in $O(\log n)$ amortized cost ($O(\log^2 n)$ worst-case cost).
\end{theorem}

\subsection{Finding Heavy Child of a Node}
\label{sec:heavy-child}

As mentioned, the strict \ourtree{} requires a building block to identify the heavy child (if any) of a given node in the \downwardcalifunc{} function.
This requires maintaining all child pointers in each node, and maintaining the heaviest child under possible changes to the tree structure. 
\ifconference{For page limit, we present the algorithm in the full version of this paper \cite{amtreefullversion}.
}  
\iffullversion{Here we discuss possible data structures to implement such queries. 
The data structure needs to support the following operations:
\begin{itemize}
  \item \funcfont{AddChild$(x, y)$}: Add $x$ as a child of $y$.
  \item \funcfont{RemoveChild$(x, y)$}: Remove $x$ from the children of $y$.
  \item \funcfont{GetHeavyChild$(y)$}: Return the heavy child of $y$, or null if $y$ does not have a heavy child.
\end{itemize}

To do this, we use bit operations to support constant time cost per operation. 
For each node $y$, we maintain the following information:
\begin{itemize}
  \item $L[0..\lfloor\log n\rfloor]$: $\lfloor\log n\rfloor + 1$ doubly linked lists.
  $L[i]$ contains $y$'s children whose subtree size is in the range $[2^i, 2^{i+1})$.
  \item $\mathit{cnt}[0..\lfloor\log n\rfloor]$: The number of children in each list.
  \item A integer $w$: the $i$'th bit of $w$ is set to 1 if $cnt[i]>0$.
\end{itemize}

When adding/removing $x$ as a child of $y$, let $i = \lfloor\log \size{x}\rfloor$.
We can simply add/remove $x$ to/from the list $L[i]$ of $y$ and update $y$'s $\mathit{cnt}[i]$ and $w$ accordingly.

For \funcfont{GetHeavyChild}$(y)$, we can directly return null if $w=0$.
Otherwise, we find the high-bit of $w$ as $h$.
If $cnt[h]=1$, we get the only child $x$ in $L[h]$.
This means that $x$ is the heaviest child of $y$. 
Therefore, we just need to check whether $\size{x}>(2/3)\size{y}$ and return the result accordingly.

The most involved case is $cnt[h]>1$,
which means $y$ has at least two children with subtree size in $[2^h, 2^{h+1})$.
In this case, we directly return null because the heaviest child cannot be greater than $(2/3)\size{y}$. 
To see why, suppose the heaviest two children are $x_1$ and $x_2$ where $2^h\le\size{x_2} \le \size{x_1}<2^{h+1}$.
Then we have $\size{x_2}/\size{x_1}>1/2$.
Then $\size{x_1}/\size{y}<\size{x_1}/(\size{x_1}+\size{x_2})=1/(1+\size{x_2}/\size{x_1})<2/3$.
Therefore $x_1$ is not the heavy child of $y$, and $y$ does not have a heavy child in this case.

All the above operations trivially have constant time cost except for
computing $\lfloor\log \size{x}\rfloor$ and taking the high-bit $h$ of $w$.
Note that $\size{x}\in[1,n]$ and $w\in[0,n]$,
so both of the operations can be addressed by preprocessing the results for all values in $[0,n]$.
In other words, we can use an array to store the $\lfloor\log k\rfloor$ and the high-bit of $k$ for each integer of $k\in [0,n]$. In this case, each time we need to compute these values, we only need $O(1)$ time lookup. Such preprocessing will take $O(n)$ time, which can be asymptotically hidden by other initialization time on arrays of size $n$ (e.g., $\parent{\cdot}$).
For the space cost of the data structure,
we do not need to create $O(\log{n})$ space for each node,
as there are only $O(n)$ children pointers in total.
We only need $O(n)$ space if we assume perfect hashing.

In practice, these two operations can be easily supported by modern CPUs. In C++, we can use \texttt{std::bit\_width} and \texttt{std::countl\_zero} to directly implement these two operations.
}  

\section{The Lazy \ourtree}\label{sec:lazy}

In Sec. \ref{sec:preemptive}, we introduced the strict version of \ourtree{}, which always keeps the tree size-balanced.
However, this version requires maintaining all the child pointers in each node and the building block in \cref{sec:heavy-child} to
identify the heavy child, which may bring up unnecessary performance overhead. 

In this section, we introduce a lazy version of \ourtree{}, which only requires each tree node to maintain the parent pointer. 
This version rebalances the tree lazily, so the tree height is not always bounded by $O(\log n)$. 
However, we will show that the lazy \ourtree{} also achieves the same $O(\log n)$ amortized cost for insertions and path-max queries. 

\subsection{Algorithms}

We present the algorithm in Alg. \ref{algo:lazy-cdt}. 
This algorithm still uses the two primitives: the same \linkfunc{} as the strict version, 
and \upwardcalifunc{}. 
Different from \downwardcalifunc{} in the strict version that rebalances a node with its children, 
the \upwardcalifunc{} function tries to rebalance a node with its \emph{parent}. 
\upwardcalifunc{}$(u)$ checks the path from $u$ to the root and ensures that
any two of $u$'s consecutive ancestors $x$ and $y=\parent{x}$ satisfy $\size{x}\le (2/3)\cdot\size{y}$. 
As such, the depth of $u$ is reset to $O(\log n)$. 
To do this, \upwardmaintain{}$(x)$ repeatedly \promote{s} $x$ if $x$ is a heavy child (i.e., its size is more than $2/3$ of its parent). 
When $x$ is no longer a heavy child, we move to its parent and continue. 

\upwardcalifunc{} balances the tree in a lazy way. 
In \lazypathmax{}$(u,v)$, we first call \upwardmaintain{} on both $u$ and $v$ to \calibrate{} the path from each of them to the root.
Then we directly use the plain algorithm to find all edges on the path and obtain the maximum one. 

In \lazyinsertfunc{}$(u,v,w)$, we also first use \upwardmaintain{} on both $u$ and $v$ to \calibrate{} the path from each of them to the root.
Then we use the same \perchandlink{} or \stitchandlink{} functions to connect $u$ and $v$ as in the strict version, 
and connect them by an edge with weight $w$ (modifying other edges of the tree if necessary). 
The algorithm does not then \calibrate{} the tree after \linkfunc{}. 
For this reason, the tree after a \lazyinsertfunc{} is not guaranteed to be size-balanced. 
The rebalance process will be postponed to the next time when a node is accessed in either an insertion or a path-max query. 

In the next section, we will show that, although the tree is not guaranteed to be balanced, the amortized costs for both \lazyinsertfunc{} and \lazypathmax{} are still $O(\log n)$.

\begin{algorithm}[t]
\small
\DontPrintSemicolon
\caption{The Lazy \ourtree{}\label{algo:lazy-cdt}}
\SetKwProg{myfunc}{Function}{}{}
\SetKwFor{parForEach}{ParallelForEach}{do}{endfor}
\SetKwFor{mystruct}{Struct}{}{}
\SetKwFor{pardo}{In Parallel:}{}{}

\myfunc{\upshape\lazyinsertfunc$(u,v,w)$}{
  \upwardmaintain$(u)$\\
  \upwardmaintain$(v)$\\
  \linkfunc$(u,v,w)$\tcp*[f]{plug in \perchandlink{} or \stitchandlink{} in Alg. \ref{algo:preemptive-cdt}}
}

\myfunc{\upshape\lazypathmax$(u,v)$}{
  \upwardmaintain$(u)$\\
  \upwardmaintain$(v)$\\
  \Return\pathmax$(u,v)$\tcp*[f]{See Alg. \ref{algo:preemptive-cdt}}
}
\medskip
\myfunc{\upshape\upwardmaintain$(x)$}{
  \While{$\parent{x}$ is not null\label{line:lazy:upmaintain:outer}}{
  \lWhile{$\size{x}>\frac{2}{3}\size{\parent{x}}$\label{line:lazy:upmaintain:inner}}{
    \promote$(x)$
  }
  $x\gets \parent{x}$\\
}
}  
\end{algorithm}

\subsection{Analysis} \label{sec:lazy:analysis}

We now analyze the lazy \ourtree{}. 
We use the same potential function as in the strict version.
We first note that the correctness of the lazy version can be directly derived from the same proof for the strict version (Thm. \ref{thm:exact:correctness}),
and the following theorem holds. 

\begin{theorem}[Correctness of the Lazy \ourtree]
\label{thm:lazy:correctness}
  Given a graph $G=(V,E)$ and a \tmst{} $T=(V,E_T)$ for $G$, after the \lazyinsertfunc{}$(u,v,w)$ in Alg. \ref{algo:lazy-cdt} using either \perchandlink{} or \stitchandlink{}, $T$ is a valid \tmst{} for $G'=(V, E\cup \{(u,v,w)\})$. 
\end{theorem}

We now analyze the amortized cost. The \upwardmaintain{}$(x)$ function will not fully \calibrate{} the tree, but it will \calibrate{} the path from $x$ to the root to ensure 
the depth of $x$ becomes $O(\log n)$, as stated in the following lemma. 

\begin{lemma}
  After \upwardmaintain{} on $u$ and $v$, the depth of $u$ and $v$ becomes $O(\log n)$.
\end{lemma}

\begin{proof}
  \upwardmaintain$(x)$ will make all nodes on the path from $x$ to the root size-balanced,
  so the depth of $u$ and then $v$ will be adjusted to $O(\log n)$. 

  However, the second \upwardmaintain{} on $v$ may affect the depth of $u$.
  Similar to the proof of \cref{lem:link-cost},
  we can show that \upwardmaintain{}$(v)$ increases the depth of any node by at most 1. 
  \upwardmaintain{}$(v)$ performs a series of \promotefunc{} operations on $v$ or $v$'s ancestors.
  Let $x$ be the node being \promote{}d and $y$ the parent of $x$,
  then the depth of nodes in $y$'s subtree may increase by 1 (the \zigzag{} case in \cref{fig:rotate}).
  After that, $y$ can never be the parent of the node being \promote{d} again.
  Hence, the depth of any node can be increased by at most 1.
  Both $u$ and $v$ have depth $O(\log n)$ after \upwardmaintain{} on $u$ and $v$.
\end{proof}

We then show that the \upwardmaintain{} function itself only has $O(\log n)$ amortized cost.
Note that since \upwardmaintain{} may work on an unbalanced tree, it may access $\Omega(\log n)$
nodes on the path, resulting in an $\Omega(\log n)$ actual cost. 
However, since some of the operations, specifically the \promotefunc{} operations,
rebalance the tree and decrement the potential function, the amortized cost can be bounded in $O(\log n)$. 

\begin{lemma}
  The amortized cost of \upwardmaintain{} is $O(\log n)$.
\end{lemma}

\begin{proof}
  First of all, note that the \promotefunc{} function in the inner while-loop on Line \ref{line:lazy:upmaintain:inner}
  is performed only if imbalance occurs. In Lemma \ref{lem:unbalancedrotate}, we proved that this operation
  has zero amortized cost, since it decrements the potential function. Therefore, the entire while-loop on Line \ref{line:lazy:upmaintain:inner}
  has $O(1)$ amortized cost, indicating that each iteration of the outer while-loop on Line \ref{line:lazy:upmaintain:outer} only has $O(1)$ amortized cost. 
  
  We then prove that the outer while-loop has $O(\log n)$ iterations. 
  This is because in each iteration, when the inner loop terminates, we must have $\size{x}\le \frac{2}{3}\size{\parent{x}}$.
  Then we update $x$ to its parent and continue to the next iteration. 
  Therefore, each iteration increases the size of the current node $x$ by at least a factor of $3/2$. 
  In at most $O(\log n)$ iterations, the outer while-loop terminates.   
\end{proof}
\vspace{-.05in}

Combining the above lemmas and the amortized cost of \perchandlink{} and \stitchandlink{} proved in Lemma \ref{lem:amortized-perch-link} and \ref{lem:amortized-stitch}, we have the following theorem.
\vspace{-.05in}

\begin{theorem}\label{thm:lazy}
  The lazy \ourtree{} supports the \lazyinsertfunc{} and \lazypathmax{} in $O(\log n)$ amortized time per operation.
\end{theorem}

\section{Persisting the \ourtree{}}
\label{sec:persistent}

We now discuss how to persist \ourtree{} upon updates, which is required in certain temporal graph applications.
Since we focus on temporal graphs, we mainly consider \emph{partial persistence}, 
where updates are applied only to the last version but we can query any history version.
The methodology here also extends to the \emph{fully persistent} setting where all versions form a DAG instead of a chain.
To persist the \ourtree{}, we need to persist the arrays $\parent{\cdot}$ and $\weight{\cdot}$. 
Below we use the $\parent{\cdot}$ as an example. 
Assume there are $m=\Omega(\log n)$ edge insertions to the \ourtree{}. 
Let $k$ be the total number of nodes that are modified by the $m$ edge insertions. The analysis in Sec. \ref{sec:preemptive:analysis:efficiency} and \ref{sec:lazy:analysis} shows that $k=O(m\log n)$.

\myparagraph{Version Lists.}
We first consider a simple and practical solution based on version lists (referred to as ``fat nodes method'' in \cite{persistence}).
All experiments in this paper use this approach.
For this approach, each node $u$ in the \ourtree{} maintains a list of versions of $\parent{u}$, in the form of $(t, \parentt{t}{u})$ ordered by $t$, where $t$ is the time when the edge is added, and $\parentt{t}{u}$ is the parent of $u$ in this version.
When the parent of $u$ is updated, a new pair of $(t,\parent{u}_t)$ is appended to the version list.
In this case, no asymptotic cost is needed for supporting persistent insertions
and the version lists take $O(k)=O(m\log n)$ space.
However, when querying a history version, we need a binary search to locate the pointers of the current version,
which adds an $O(\log k)=O(\log m)$ overhead to query costs.

Note that only the strict \ourtree{} can guarantee the $O(\log m\log n)$ query cost,
where \pathmax{} will check $O(\log n)$ edges in the tree.
The query cost for the lazy \ourtree{} is amortized; however, in practice, the difference in query performance between the two versions is minimal.

\myparagraph{vEB-Trees-based Solution.}
Theoretically, the overhead for persistence can be reduced from $O(\log m)$ to $O(\log\log m)$ by using the approach given in \citet{straka2009optimal}.
At a high level, the ordered set is maintained by a van Emde Boas Tree~\cite{van1977preserving} that provides doubly logarithmic update and lookup cost. 

\section{Applications on Temporal Graphs}\label{sec:app}

Given the algorithms for \ourtree{} with support for persistence, we are ready to solve various temporal graph applications.
In \cref{sec:prelim:temporal}, we briefly introduced the point-interval temporal connectivity problem.
In this section, we show other temporal graph problems and how \ourtree{s} can solve them.

\vspace{-.1in}

\subsection{Temporal Graph Settings}

Two categories in temporal graph processing have received significant attention. 
The first is the \emph{point-interval setting} (e.g.,~\cite{ciaperoni2020relevance,li2018persistent,anderson2020work,crouch2013dynamic,song2024querying,tian2024efficient,xie2023querying,yang2023scalable,yu2021querying,yang2024evolution,zhang2024incremental}) as mentioned in \cref{sec:prelim:temporal}.
In this setting, each edge $e$ has a timestamp $t(e)$ (i.e., edge $(u,v)$ arrives at time $t(e)$).
A query is associated with a time interval $[t_1, t_2]$
and is performed on a sub-graph $G'_{[t_1,t_2]}$ with edge set $E'=\{e~|~t(e)\in[t_1,t_2]\}$.
A simpler case is the so-called \emph{sliding-window setting}~\cite{anderson2020work,crouch2013dynamic}. 

Dually, there is the \emph{interval-point setting}
(e.g.,~\cite{gandhi2020interval,da2024kairos,peng2019optimal,eppstein1994offline,holme2013epidemiologically,rocha2013bursts,bearman2004chains}), where
each edge $e$ has a time interval $[t_1(e), t_2(e)]$. 
A query is associated with a timestamp $t$ and is performed on the sub-graph $G'_t$ with edge set $E'=\{e~|~t\in[t_1(e),t_2(e)]\}$.
A similar setting is the ``offline dynamic graphs'' \cite{peng2019optimal,eppstein1994offline}, where each edge can be inserted/deleted at a certain time, and queries are performed on a snapshot of the graph. 
From a temporal view, each edge has a lifespan (an interval) from its insertion to its deletion, and queries are performed on a specific timestamp.
However, in fully dynamic graphs \cite{hanauer2021recent,pandey2021terrace,henzinger2020dynamic,bisenius2018computing,tench2024graphzeppelin,holm2001poly,mccoll2013new,liu2022parallel}, the deletion time is unknown at the time of insertion.

\vspace{-.1in}
\subsection{Online/Offline Settings}

The graph and queries can also be either online or offline.
Offline means the information is known ahead of time, while online means the algorithm needs to respond to every update/query before the next one comes.
We first consider the graph:
\begin{itemize}[nosep]
  \item \emph{Offline Graph}~\cite{tian2024efficient,xie2023querying,yang2023scalable,yu2021querying,yang2024evolution,peng2019optimal,eppstein1994offline}: all edges in the graph are known ahead of time (before or with the queries).  
  \item \emph{(Online) Streaming Graph}~\cite{anderson2020work,crouch2013dynamic,song2024querying,tian2024efficient,zhang2024incremental}: New edges arrive one by one, forming a graph stream.  
      In this case, the timestamp of the edge is the ordering of it in the stream, so only the point-interval setting applies here.
\end{itemize}

Note that the offline setting can be converted to online by sorting all edges based on the time and processing them.

The queries can also come in different settings:
\begin{itemize}[nosep]
  \item \emph{Offline} \cite{song2024querying,anderson2020work,crouch2013dynamic,yu2021querying,zhang2024incremental}: Queries are known ahead of time.
  \item \emph{Historical (Online)} \cite{tian2024efficient,xie2023querying,yang2023scalable,yu2021querying,yang2024evolution,peng2019optimal,eppstein1994offline}: Queries come as a stream, and can travel back in history to query any previous timestamp or time interval. This requires to persist the graph (or the corresponding data structure). 
\end{itemize}


In summary, there are a variety of different temporal graph settings,
and they have been studied either within the temporal graph scope or as other problems (e.g., offline dynamic graphs~\cite{peng2019optimal,eppstein1994offline}). 
However, even though the literature has designed solutions for some specific settings, one contribution of our work is to show how a base data structure
can be adapted to different settings.
In particular, the \ourtree{}, which supports efficient incremental MST, can be used for a wide range of problems (mostly connectivity-related problems) in this section, combined with all the settings discussed above.
Next, we will use connectivity as the main example, and show two other problems that can also be solved with some moderate modifications.
\ifconference{For page limit, more applications are discussed in the full version of this paper \cite{amtreefullversion}.}
For many applications, their reductions to MST-related problems have been studied in a specific graph-query setting~\cite{anderson2020work,crouch2013dynamic}. 
Our discussions show that they can all be solved by \ourtree{s} and can be extended to other settings in a straightforward way.

\subsection{Connectivity}\label{sec:connectivity}

On an undirected graph $G=(V,E)$ there are two crucial problems related to graph connectivity:
\begin{itemize}[nosep]
    \item Determine whether $u$ and $v$ are connected in $G$.
    \item Report the number of connected components in $G$.
\end{itemize}
In temporal graph applications, the graph contains edges with temporal information.
We show that both the point-interval and the interval-point settings can be converted to incremental MST and solved by \ourtree{s} efficiently.

The point-interval setting is discussed in \cref{sec:prelim:temporal} as a motivating example and we briefly recap here.
Each edge $e$ is treated as an edge insertion at time $t(e)$ with weight $-t(e)$.
We can then maintain an \ourtree{} by processing all edges in order as an incremental MST. 
We use $T_t$ to denote the \ourtree{} up to time $t$.
For a query $(u,v,t_1,t_2)$, we check and report if the $\pathmax{}_{T_{t_2}}(u,v)=w$ satisfies $|w|\ge t_1$~\cite{crouch2013dynamic,anderson2020work,song2024querying}.
To report the number of connected components (CC)~\cite{crouch2013dynamic},
note that all edges $e$ in $T_{t_2}$ with $t(e)<t_1$ break the connectivity of the graph and increase the number of CCs by 1.
We keep an ordered set $D$ to store all edges in the current MST, ordered by $t(e)$.
For each edge insertion, we update the edges in $D$ (up to one edge inserted/removed).
For a query $(t_1,t_2)$, we look at $D_{t_2}$ ($D$ at time up to $t_2$), and report $n-|\{e~|~t(e)\ge t_1\}|$.
$D$ can be maintained by any balanced BST in $O(\log n)$ cost per operation.

For the interval-point setting, each edge has a time interval $[t_1(e),t_2(e)]$. 
We convert it to an incremental MST problem by adding this edge at time $t_1(e)$ with weight $-t_2(e)$. 
Again we use $T_t$ to denote the \ourtree{} up to time $t$.
For a query $(u,v,t)$, we query $w=\pathmax{}_{T_t}(u,v)$ on $T_t$ and check whether $|w|\ge t$.
If so, $u$ and $v$ remain connected at time $t$. 
We can use \ourtree{} for the number of connected components queries similarly.

Note that we need to perform the \pathmax{} (or check BSTs for the number of CCs) for each query.
If the queries are offline, we can sort the query time ($t$ or $t_2$) together with the edges, so all \pathmax{} queries apply to the ``current'' \ourtree{} in the stream.
For the historical setting, we need to persist the \ourtree{} (and also the ordered set $D$), so queries can travel back and check any previous version of \ourtree{} or $D$.
We show the theoretical guarantees on this problem along with the following application on bipartiteness in Thm. \ref{thm:apps}.

\subsection{Bipartiteness}\label{sec:bipartite}

An undirected graph $G=(V,E)$ is bipartite iff there exists a vertex subset $V'\subseteq V$
such that every edge has one endpoint in $V'$ and the other endpoint in $V\setminus V'$.

There is a known reduction \cite{ahn2012analyzing,crouch2013dynamic} of the bipartiteness problem to the connectivity problem.
We generate another graph $G'$ by duplicating each vertex $v\in V$ into two copies $v_1$ and $v_2$ in $G'$, and duplicating each edge $(u,v)\in E$ into $(u_1, v_2)$ and $(v_1, u_2)$ in $G'$.
The graph $G$ is bipartite if and only if $G'$ has twice as many connected component as $G$.

Bipartiteness checking in the temporal setting is similar to connectivity.
We run the same algorithm for connectivity on both $G$ and $G'$.
For a query at time $t$, we check and return if the number of connected components on $G'_t$ is twice as $G_t$.
The same cost analysis for connectivity also applies here. 
Using vEB tree for persistence leads to the following theorem.

\begin{theorem}\label{thm:apps}
  Given a temporal graph with $n$ vertices and $m$ edges, the temporal connectivity or bipartiteness can be solved by \ourtree{s} with $O(n)$ initialization cost and $O(\log n)$ cost per edge update; 
  the offline query and historical query have $O(\log n)$ and $O(\log n \log \log m)$ cost, respectively.
\end{theorem}

Note that we assume $m=\Omega(n)$ since otherwise singleton vertices can be filtered out.

\subsection{$k$-Connectivity and $k$-Certificate}\label{sec:k-certificate}

Given an undirected graph $G=(V,E)$, two vertices $u$ and $v$ are $k$-connected if there are $k$ edge-disjoint paths connecting them.
A graph is $k$-connected if every pair of vertices is $k$-connected.

A $k$-certificate is a sequence of edge-disjoint spanning forests $F_1, F_2, ..., F_k$ from $G$,
and $F_i$ is a maximal spanning forest of $G\setminus (F_1 \cup F_2 \cup \cdots \cup F_{i-1})$.
The connection between the $k$-certificate and $k$-connectivity is that
$u$ and $v$ are $k$-connected in $G$ if and only if they are $k$-connected in $(F_1 \cup F_2 \cup \cdots \cup F_{k})$.

We can generate $k$-certificate using a connectivity algorithm~\cite{crouch2013dynamic}.
$F_1$ is simply the same MST computed in \cref{sec:connectivity} using the \ourtree{}.
Then, when $F_i$ is updated---an edge $e$ is replaced by another edge in the MST, it will be inserted into $F_{i+1}$.
Hence, in total we maintain $k$ \ourtree{s}, so the cost is multiplied by $k$ (asymptotically the same when assuming $k=O(1)$).

\iffullversion{\subsection{Approximate MSF Weight}

If the edge weights are between $1$ and $n^{O(1)}$,
there exists a known reduction \cite{ahn2012analyzing,chazelle2005siam} to approximate the MSF weight within a factor of $1+\epsilon$
by tracking the number of connected components in graphs $G_0, G_1, ..., G_{R-1}$,
where $G_i$ is a subgraph of $G$ containing all edges with weight at most $(1+\epsilon)^i$ and $R=O(\epsilon^{-1}\log n)$.
Specifically, the MSF weight is given by
\begin{equation}
  n-cc(G_0)+\sum_{i=1}^{R-1}(cc(G_i)-cc(G_{i-1}))(1+\epsilon)^i.
\end{equation}

In the temporal setting,
using the same techniques introduced in \cref{sec:connectivity},
we can use $R$ instances of \ourtree{} to track the number of connected components in $G_0, G_1, ..., G_{R-1}$.

\subsection{Cycle-freeness}

The cycle-freeness problem asks to determine whether a graph contains a cycle,
i.e., whether there exists a path $v_1, v_2, ..., v_k$ such that $v_1=v_k$ and $v_i\neq v_j$ for all $i\neq j$.

Note that in a cycle-free graph,
any pair of nodes $(u,v)$ is not biconnected.
Thus, we can use the $k$-certificate algorithm \cref{sec:k-certificate} with $k=2$ to solve this problem.
To determine whether a graph contains a cycle,
we can simply check whether $F_2$ contains any edge.
Because we only need to maintain $F_1$ and $F_2$,
our algorithm gives the same bound as Thm. \ref{thm:apps}.
}  

\ifconference{\subsection{Other Applications}

Due to the space limit, we discuss other applications in the full version of this paper \cite{amtreefullversion}.
}

\section{Experiments}\label{sec:exp}

This section provides experimental evaluation of the effectiveness of \ourtree{s}.
We mostly focus on one setting, the point-interval temporal connectivity, due to the following reasons.
First, there exist fast baselines for this problem~\cite{song2024querying} that are apple-to-apple comparisons to \ourtree{s}.
Second, when mapping to incremental MST, the interval-point setting only changes the edge weight distribution, and the runtime is similar.
\ifconference{Additional experiments are in the full version of this paper \cite{amtreefullversion}.}

For the point-interval connectivity, each edge $e$ has a timestamp $t_e$. 
A query has a time interval $[t_1,t_2]$,
and only edges $e$ with timestamp $t_e\in[t_1,t_2]$ are considered in the query. 
In this section, we mainly focus on querying the connectivity between two vertices.
We provide the results for querying the number of connected components in \iffullversion{\cref{sec:exp-count-cc}}\ifconference{the full version of this paper \cite{amtreefullversion}}.
As discussed in Sec.~\ref{sec:connectivity}, it is an important building block for many temporal applications such as bipartiteness checking. 
Our source code is publicly available on Github \cite{amtreecode}.


\subsection{Setup}
We implemented the strict and the lazy versions of \ourtree{} in C++ and persist them by version lists (see Sec. \ref{sec:persistent}). 
Experiments are run on a Linux server with four Intel Xeon Gold 6252 CPUs and 1.5 TB RAM, though only one core is utilized.
We compiled our code using Clang 18.1 with \texttt{-O3} flag.

\myparagraph{Datasets}
We tested seven real-world graphs (summarized in \cref{tab:graph-info}) with very different features. 
The first three graphs are real-world temporal graphs where each edge is associated with a timestamp.
The last four are static graphs and we assign a uniformly random timestamp in the range $[0,10^9]$ to each edge.


\hide{
\begin{table}[h]
\centering
\caption{Graph Information Table\label{tab:graph-info}}
\begin{tabular}{lcrr}
    \toprule
    Name & Temporal? & $|V|$ & $|E|$ \\
    \midrule
    wiki-talk \cite{snapnets} & Yes & 1,140,149 & 7,833,140  \\
    sx-stackoverflow \cite{snapnets} & Yes & 6,024,271 & 63,497,050  \\
    soc-bitcoin \cite{nr} & Yes & 24,575,383  & 122,378,012  \\
    RoadUSA \cite{roadgraph} & No & 23,947,348 & 57,708,624  \\
    GeoLife \cite{geolife,wang2021geograph} & No & 24,876,978 & 124,384,890  \\
    Twitter \cite{kwak2010twitter} & No & 41,652,231 & 1,468,365,182  \\
    sd\_arc \cite{webgraph} & No & 89,247,739 & 2,043,203,933  \\
    \bottomrule
\end{tabular}
\end{table}
}

\newcommand{\graphname}[1]{\textsf{#1}}
\newcommand{\WT}{\graphname{WT}}
\newcommand{\SX}{\graphname{SX}}
\newcommand{\SB}{\graphname{SB}}
\newcommand{\USA}{\graphname{USA}}
\newcommand{\GL}{\graphname{GL5}}
\newcommand{\TW}{\graphname{TW}}
\newcommand{\SD}{\graphname{SD}}

\begin{table}[t]
\centering
\small
\vspace{-.5em}
\begin{tabular}{llrr}
    \toprule
    Name & Graph &  $|V|$ & $|E|$ \\
    \midrule
    \WT{}&$^*$wiki-talk \cite{snapnets} &  1.1M & 7.8M  \\
    \SX{}&$^*$sx-stackoverflow \cite{snapnets}  & 6.0M & 63.5M  \\
    \SB{}&$^*$soc-bitcoin \cite{nr} &  24.6M  & 122.4M  \\
    \USA{}&RoadUSA \cite{roadgraph} &  24.0M & 57.7M  \\
    \GL{}&GeoLife \cite{geolife,wang2021geograph}  & 24.9M & 124.3M  \\
    \TW{}&Twitter \cite{kwak2010twitter}  & 41.7M & 1.47B  \\
    \SD{}&sd\_arc \cite{webgraph}  & 89.2M & 2.04B  \\
    \bottomrule
\end{tabular}
\caption{Graph Information. $^*$: real-world temporal graphs. Others are static graphs with randomly generated temporal information.  \label{tab:graph-info}}
\end{table}

\newcommand{\impname}[1]{\textsf{#1}}
\newcommand{\strictstitch}{\impname{Strict-Stitch}}
\newcommand{\strictperch}{\impname{Strict-Perch}}
\newcommand{\lazyperch}{\impname{Lazy-Perch}}
\newcommand{\lazystitch}{\impname{Lazy-Stitch}}
\newcommand{\oec}{\impname{OEC-Forest}}
\newcommand{\linkcut}{\impname{LC-Tree}}

\myparagraph{Evaluated Methods}
We compared six data structures in total. 
For each of them, we test the throughput for both \emp{updates} (processing all temporal edges) and \emp{queries}.
\begin{itemize}[noitemsep]
  \item \textbf{\strictstitch{}, \strictperch{}, \lazystitch{}, \lazyperch{}}: Our implementations of four versions of \ourtree{} using strict/lazy strategy based on \perchfunc{}/\stitchfunc{}. 
  \item \textbf{\oec{}}~\cite{song2024querying}: A state-of-the-art implementation for incremental MST, which solves temporal connectivity. 
  \item \textbf{\linkcut{}}: Our own implementation of link-cut trees \cite{sleator1983data}. 
\end{itemize}

Recall that the \linkcut{} is a classic data structure offering theoretical guarantees, whereas \oec{} is a practical data structure
without non-trivial bounds. All four versions of \ourtree{} provide the same (amortized) bounds as \linkcut{}, and are also designed to be practical. For the \ourtree{s} and \oec{} we also tested their persistent versions for historical queries. 
We note that, as mentioned in Sec. \ref{sec:persistent}, the lazy \ourtree{s} do not guarantee the polylogarithmic query bound. 
The update bounds for the lazy version, and all bounds for the strict versions still hold in the persistent setting. 

\begin{figure*}[t]
  \centering
  \vspace{-1em}
  \includegraphics[width=0.8\textwidth]{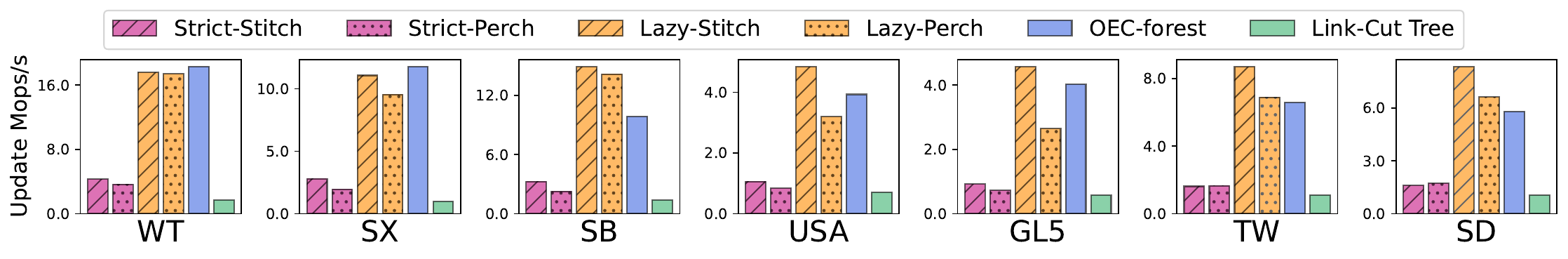}
  \includegraphics[width=0.8\textwidth]{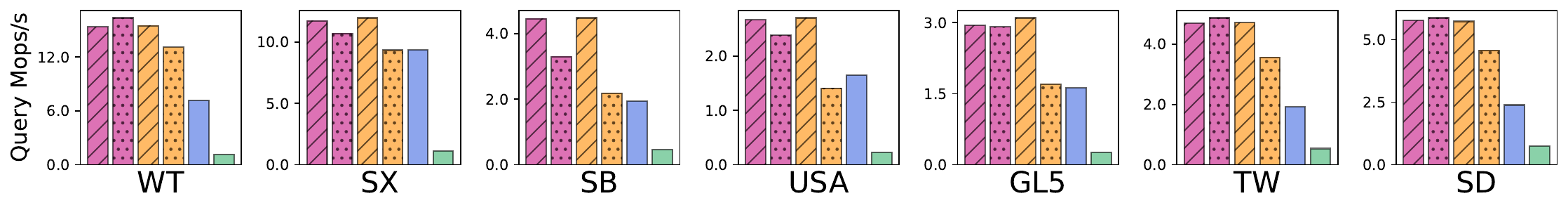}
  \caption{Update and query throughput (millions of operations per second) for \textbf{offline queries}. Higher is better.
  \label{fig:streaming_query}}
  \includegraphics[width=0.8\textwidth]{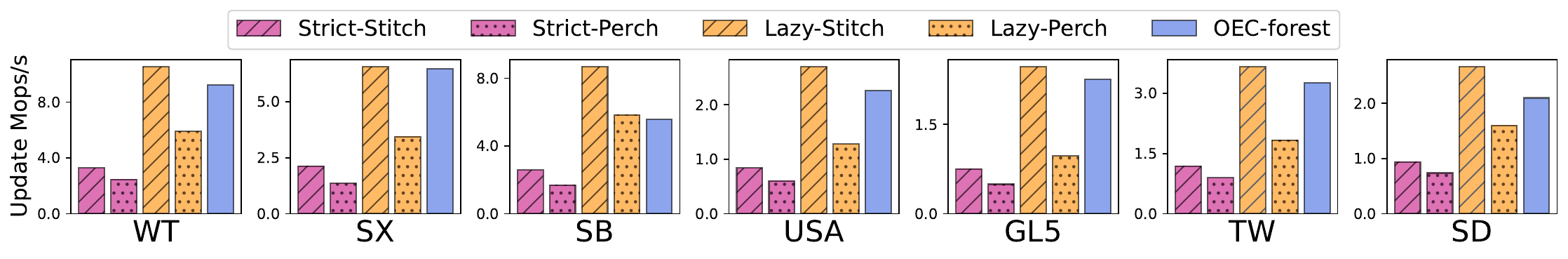}
  \includegraphics[width=0.8\textwidth]{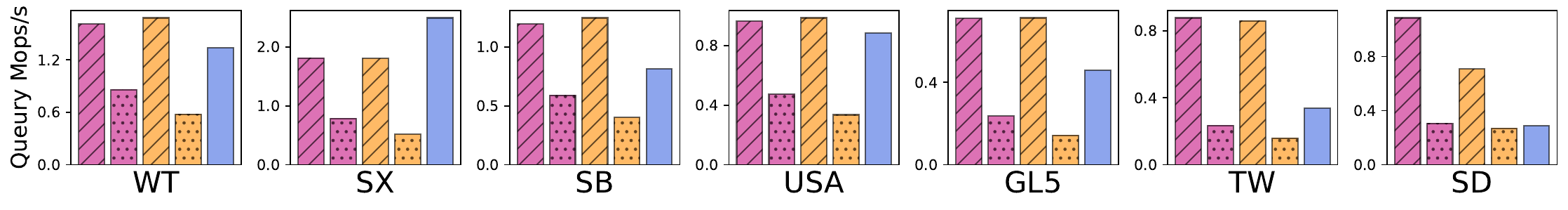}
  \caption{Update and query throughput (millions of operations per second) for \textbf{online (historical) queries}. Higher is better.\label{fig:persistent_query}}
  \vspace{-.75em}
\end{figure*}

\subsection{\ourtree{s} for Offline Queries}

We first tested the non-persistent \ourtree{} for offline queries, i.e., the queries are given ahead of time with all edges. 
In this case, there is no need to persist the \ourtree{}. We can simply process (insert) the edges in order, and after each insertion, if there is a query that corresponds to this time, we directly perform it. 
\cref{fig:streaming_query} shows the update and query throughput in this setting. 

\myparagraph{Update Throughput.} We first compare among the four versions of \ourtree{} in updates. The lazy version always achieves much better performance than the strict version, due to two main reasons.
First, the lazy version does not maintain the children pointers and does not actively check the heaviest child, which saves much work.
Second, the lazy version does not rebalance the whole tree after an update, and thus requires less work than the strict version. 
In total, the performance for the lazy version is 3.6--6.2$\times$ faster on average on all graphs. 

The \stitch{}-based versions are usually slightly faster than the \perch-based versions. 
Such a difference is more pronounced in the persistent settings, which we discuss later.

Compared to other baselines, while \linkcut{} achieves strong theoretical guarantee, it has the lowest throughput on all graphs. 
It is slower than the strict \ourtree{s} by a factor of 1.2--2.6$\times$, 
and is slower than the lazy \ourtree{s} and \oec{} by at least 4.5$\times$. 
\oec{} tree has reasonably good performance on all graphs. 
The best version of \ourtree{s}, \lazystitch{} still achieves competitive or better performance than \oec{},
which is from 4\% slower (on \WT{}) to 1.5$\times$ faster (on \SB{}). On average across seven graphs, \lazystitch{} is 1.2$\times$ faster. 
This speedup comes from the theoretical guarantees of the \ourtree{} that leads to shallower tree depths. 

\myparagraph{Query Throughput.} For queries, all versions of \ourtree{} have better performance than both \oec{} and \linkcut{}.
The advantage over \linkcut{} is from the algorithmic simplicity,
and the advantage over \oec{} is from the depth guarantee of \ourtree{} in theory. 
To verify this, we further tested the average tree height for \ourtree{} and \oec{}, and present the results in \iffullversion{\cref{sec:exp-tree-height}}\ifconference{the full version of this paper \cite{amtreefullversion}} for completeness.
Comparing \oec{} with \lazystitch{} as an example, \oec{} is 1.8--2.9$\times$ deeper than \ourtree{}, 
making \ourtree{} 1.6--2.5$\times$ faster than \oec{} for queries.


\hide{
\begin{figure}[t]
  \centering
  \includegraphics[width=1.0\columnwidth]{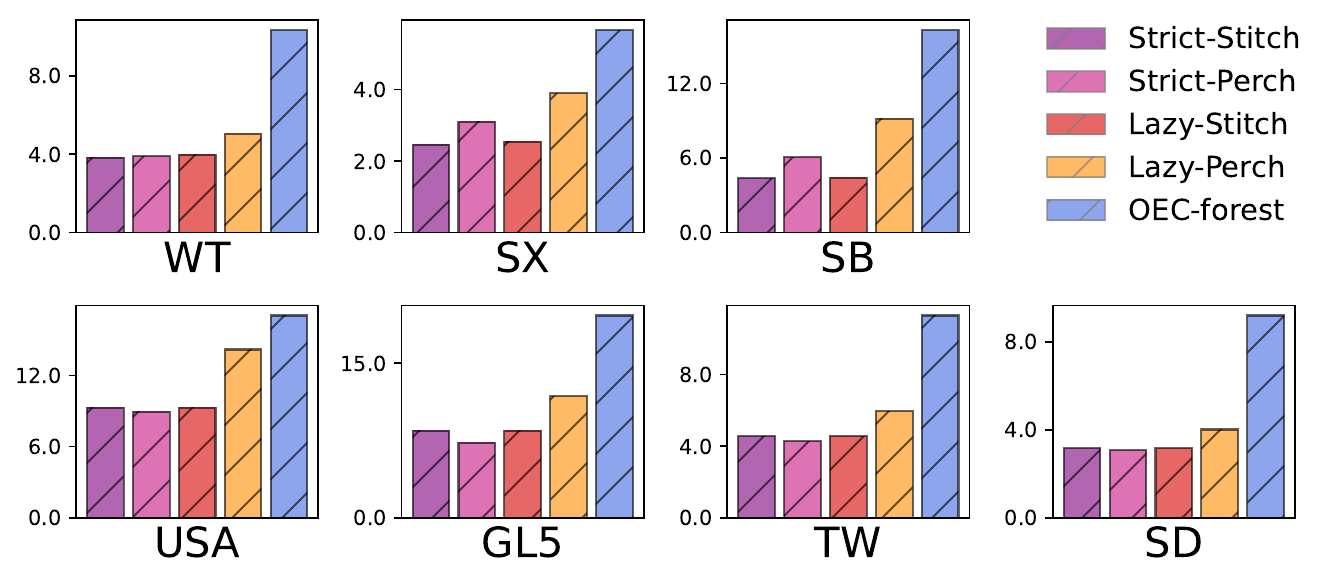}
  \caption{Average Heights.\label{fig:avg_h}}
\end{figure}
}

\subsection{\ourtree{}s for Historical Queries}

We now discuss the setting with historical queries, which requires using the persistent version of \ourtree{s}. 
In this setting, the queries are not known when the index is constructed, so we need to preserve all versions of the \ourtree{} at all timestamps. 
We present the results in \cref{fig:persistent_query}. 

The performance for updates is pretty consistent with the non-persistent version. In all cases, \lazystitch{} achieves the best performance,
and \oec{} is close to our best performance. 
For queries, the slowdown of the \perch{}-based version over the \stitch{}-based one becomes significant. 
As mentioned, the difference comes from the more substantial tree restructuring in \perchfunc{}. 
\perchandlink{} changes $\Theta(d(u)+d(v))$ nodes in the tree. Note that this bound is tight, since $u$ and $v$ both have to be perched to the top,
causing all nodes on the path to generate a new version. 
For \stitchandlink{}, in many cases, the edge is just conceptually moved up without changing the tree. 
To verify this, in \iffullversion{\cref{sec:version-lists-size}}\ifconference{the full version of this paper \cite{amtreefullversion}} we report the number of versions generated during the algorithm, which indicates the total number of nodes that have been touched and changed their parent/child pointers during the entire algorithm. The \perch{}-based algorithms indeed modified 1.4--5.5$\times$ more nodes than the \stitch{}-based versions. 

Since the lazy versions have loose query bounds, the strict version achieves better performance than the lazy ones. This is more pronounced for the \perch{}-based algorithms. For the \stitch{}-based algorithms, the difference is marginal except for the last graph \SD{}. On all graphs other than \SX{}, both \strictstitch{} and \lazystitch{} outperforms the baseline \oec{}. 

\smallskip{}

In summary, \lazystitch{} achieves the best overall performance for almost all settings. 
When the application emphasizes the query throughput in the online setting, 
\strictstitch{} may provide better performance in queries.

\iffullversion{\subsection{Tree Height Comparison}\label{sec:exp-tree-height}
To further study the performance gain of \ourtree{}, we tested the tree height for all four versions of \ourtree{} with \oec{}, in the non-persistent setting. The results are presented in Tab. \ref{tab:avg_h}. 
In general, the difference in tree height is highly consistent in the query performance of different data structures. 
All versions of \ourtree{} guarantees size-balance invariant, and thus a low height (although maintained lazily in the lazy version),
while \oec{} do not have non-trivial guarantee in tree height. 
Among them, the lowest tree height is usually achieved by the strict versions, due to rebalancing immediately after each insertion. 
\lazystitch{} also has similar tree height to the lowest, and \lazyperch{} can be off by 1.3--2.1$\times$.
For \oec{}, due to the lack of theoretical guarantee, the tree height can be 1.9--3.7$\times$ larger than the best, leading to the same order of magnitude of slowdown in query time. 

Another interesting finding is that in all the cases, the height of \ourtree{s} can be much lower than $\log n$ in practice. For \strictstitch{}, \strictperch{}, and \lazystitch{}, the tree height is within 10 for all the seven graphs with up to 90M vertices.

\begin{table}[t]
  \centering
  \setlength{\tabcolsep}{3.5pt}  
  \begin{tabular}{@{}lrrrrrrr}  
      \toprule
      & \WT & \SX & \SB & \USA & \GL & \TW & \SD \\  
      \midrule
      \strictstitch{}  & \textbf{3.83}&	\textbf{2.46}&	\textbf{4.39}&	9.27&	8.46&	4.54&	3.16  \\
      \strictperch{}  & 3.90&	3.10&	6.07&	\textbf{8.95}&	\textbf{7.28}&	\textbf{4.28}&	\textbf{3.07}   \\
      \lazystitch{}        & 3.96&	2.54&	\textbf{4.39}&	9.28&	8.46&	4.54&	3.16 \\
      \lazyperch{}        & 5.05&	3.90&	9.15&	14.19&	11.80&	5.95&	4.02 \\
      \oec{}         & 10.32&	5.66&	16.26&	17.05&	19.63&	11.29&	9.20   \\
      \bottomrule
  \end{tabular}
  \caption{The average height for the tested data structures. The lowest height is highlighted. }
  \label{tab:avg_h}
\end{table}

\subsection{Number of Versions in the Persistent Setting}\label{sec:version-lists-size}
To illustrate the overhead of persisting different data structures,
in \cref{tab:avg_h} we report the number of versions generated in the experiment for four versions of \ourtree{s} and compare it with \oec{}. 
This indicates the total number of tree nodes touched/updated during the entire algorithm, as well as the total memory usage. 

The two \stitch{}-based versions generate the fewest number of versions, and the \oec{} may generate 1.1--1.5$\times$ more versions than them.
The \perch{}-based versions, however, can result in up to 5.5$\times$ more versions. This illustrates that the \perchfunc{} algorithms restructure the tree more substantially.

\begin{table}[t]
  \centering
  \setlength{\tabcolsep}{3.5pt}  
  \begin{tabular}{lrrrrrrr}  
      \toprule
      & \WT & \SX & \SB & \USA & \GL & \TW & \SD \\  
      \midrule
      \strictstitch{}  & \textbf{8.8} &  83.9 &  \textbf{113} &  78 & \textbf{150} &\textbf{1,498}& \textbf{2,079}\\
      \strictperch{}  & 21.9 & 204 & 167&  222&  632 &4,909& 5,757\\
      \lazystitch{}   & \textbf{8.8}&   \textbf{83.8} & \textbf{113} &  \textbf{77} & \textbf{150} &\textbf{1,498}&  \textbf{2,079}\\
      \lazyperch{}   & 28.2 & 260 & 196 & 254 & 828 &5,750 & 7,159\\
      \oec{}     & 11.4 &  93.5 & 146 & 101 & 223 &2,244 &3,109\\
      \bottomrule
  \end{tabular}
  \caption{Millions of Updates in the Version Lists.}
  \label{tab:version-lists-size}
\end{table}

\hide{
\begin{table}[t]
  \centering
  \setlength{\tabcolsep}{3.5pt}  
  \begin{tabular}{lrrrrrrr}  
      \toprule
      & WT & SX & SB & USA & GL5 & TW & SD \\  
      \midrule
      Strict-Stitch  & \textbf{33}&  320&  433&  296&  573&   \textbf{5,717}& \textbf{7,931}\\
      Strict-Perch  & 84&  779&  638&  846&  2,411& 18,727& 21,962\\
      Lazy-Stitch        & \textbf{33}&  \textbf{319}&  \textbf{432}&  \textbf{295}&  \textbf{572}&   \textbf{5,717}& \textbf{7,931}\\
      Lazy-Perch        & 107&   994&  749&  970&  3,158& 21,938& 27,310\\
      OEC-forest         & 43&  357&  558&  385&  852&   8,564& 11,861\\
      \bottomrule
  \end{tabular}
  \caption{Memory (MB) of Version Lists}
  \label{tab:version-lists-size}
\end{table}
}

\subsection{Reporting the Number of Connected Components}\label{sec:exp-count-cc}

We also test another setting using \ourtree{}, where queries ask for the number of connected components in the graph.
As mentioned in \cref{sec:connectivity},
we need to maintain an ordered set of the current MST edges to answer the number of connected components. 
To do this, we implement this by using the PAM library~\cite{sun2018pam},
which supports the construction of persistent ordered sets in parallel. 
In our implementation, we record the edge inserted/deleted to the MST in a list,
and at the end of the build we construct all versions of the ordered set $D$ in parallel using PAM. 
We report the performance of our implementation for the offline and historical settings in \cref{fig:streaming-count-cc,fig:persistent-count-cc}, respectively. 
We only run the experiment on the first five graphs because the persistent ordered sets for \TW{} and \SD{} are too large to fit in the memory.

In both figures, we separate the cost of maintaining the ordered sets to understand the overhead for this part. 
In the offline query setting, for all versions of our algorithm, the overhead is at most 10\%, which is negligible.
In the historical query setting, the overhead varies from 20\% to 100\%,
but when the graph is large (e.g., \USA{} and \GL{}), the overhead is within 50\%.

\begin{figure*}[t]
  \centering
  \includegraphics[width=0.6\textwidth]{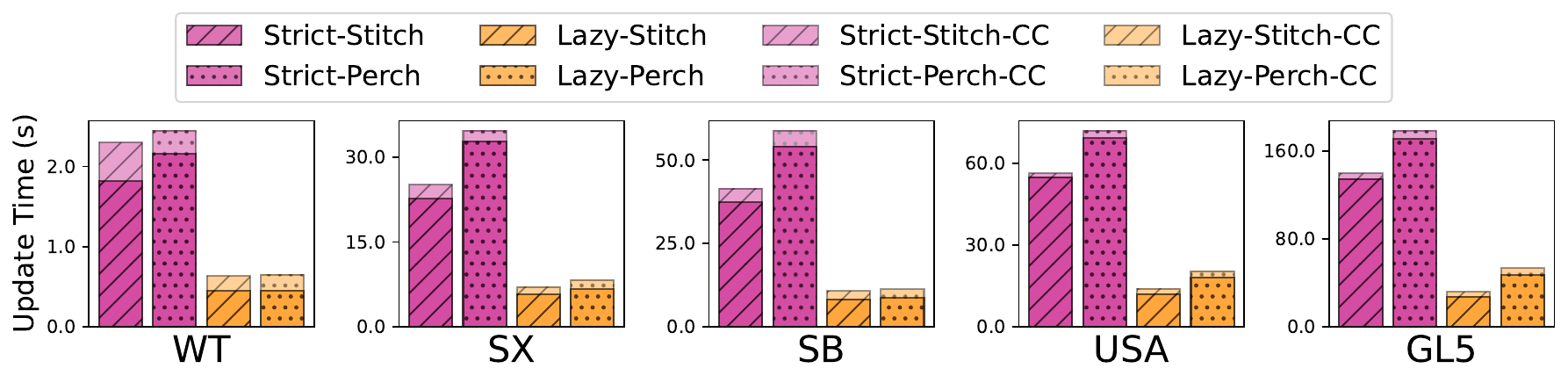}
  \caption{Overhead of maintaining the ordered set for offline queries.\label{fig:streaming-count-cc}}
  \includegraphics[width=0.6\textwidth]{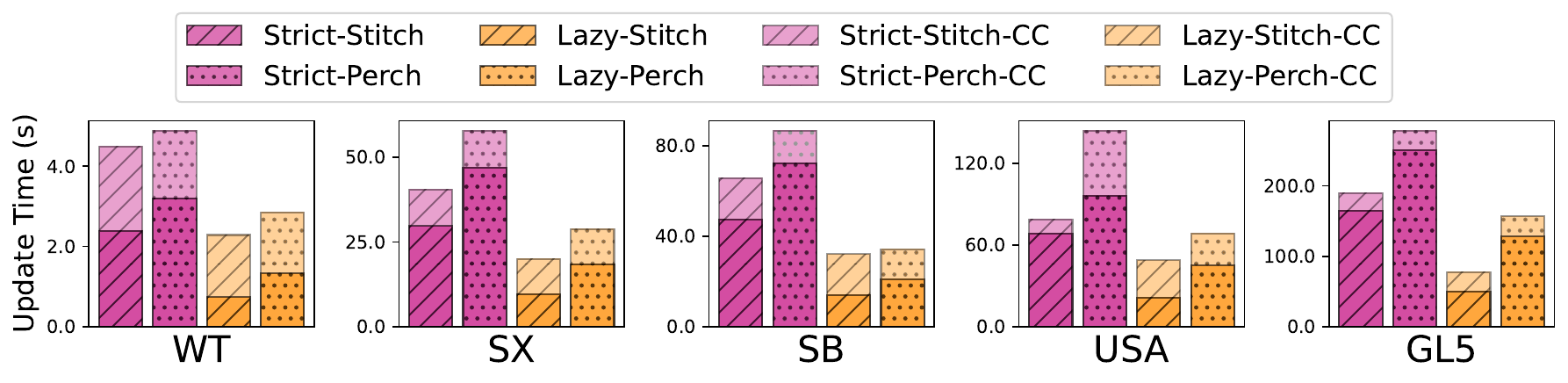}
  \caption{Overhead of maintaining the ordered set for historical queries.\label{fig:persistent-count-cc}}
\end{figure*}
}  

\section{Related Work} 

Minimum spanning tree/forest (MST/MSF) is one of the most fundamental graph problems, and has been studied from a century ago~\cite{jarnik1930jistem,boruvka1926jistem} to recent years \cite{khan2012fast,jayaram2024massively,dhouib2024innovative}.
Some famous algorithms include but are not limited to: Bor\r{u}vka's algorithm~\cite{boruvka1926jistem}, Prim's algorithm~\cite{prim1957shortest,jarnik1930jistem}, Kruskal's algorithm~\cite{kruskal1956shortest}, and KKT algorithm~\cite{karger1995randomized}.
Regarding dynamic MSTs, the classic setting with both edge insertions and deletions is challenging---the best-known algorithm~\cite{holm2015faster} needs $O(\log^4n/\log\log n)$ amortized cost per edge update.
Incremental MST with only edge insertions is simpler and is proven to be very useful.

Some classic data structures solve incremental MST efficiently in theory, including the link-cut tree~\cite{sleator1983data}, the rake compress tree (RC-tree)~\cite{acar2005experimental}, and the top tree~\cite{tarjan2005self}.
They can support each edge insertion in $O(\log n)$ cost either amortized or on average.
These data structures actually solve the more general ``dynamic tree/forest'' problem (see~\cite{acar2020changeprop}).
There are also parallel algorithms that apply 
a large batch of edge updates~\cite{anderson2020work,ferragina1996three,pawagi1993optimal,shen1993parallel}.
To the best of our knowledge, these results are mostly of theoretical interest and no implementations are available.
Practically, people have designed data structures such as the OEC-forest~\cite{song2024querying} and the D-tree~\cite{chen2022dynamic} for faster performance.
D-tree maintains a BFS-tree and patches it when updates come.
It has decent performance when the graph has certain properties, but no non-trivial cost bounds are known. 
The OEC-forest~\cite{song2024querying} was the latest work on this topic and also the main baseline we compare with. 
The OEC-forest is a \tmst{} using an idea similar to our \stitch{}-based algorithms. 
However, it does not support any non-trivial (better than linear) bounds for the tree diameter. 
Our main improvement is to introduce the anti-monopoly rule, which bounds the tree height and guarantees the cost bounds for \ourtree{}. 

Temporal graph processing is a popular research topic recently, and we refer the audience to an excellent survey~\cite{holme2012temporal} for more background.
The connection between temporal graph and incremental MST has been shown, but only for specific cases.
Song et al.~\cite{song2024querying} discussed the historical point-interval connectivity, and Anderson et al.~\cite{anderson2020work} discussed the offline point-interval setting.
To the best of our knowledge, the generalization of this connection is novel in our paper.

\section{Conclusion}

This paper proposes new algorithms for incremental MST for efficient temporal graph processing on numerous applications. 
Our new data structure, the \ourtree{}, is efficient both in theory and in practice.
In theory, the cost bounds of using \ourtree{s} for temporal graphs match the best-known results using link-cut trees or other data structures. 
In practice, we compare \ourtree{} to both the theoretically-efficient solution and state-of-the-art practical solutions. 
Our \lazystitch{} version achieves the best performance in most tests including various graphs with offline/historical queries on both updates and queries. 


\hide{
\section*{Acknowledgement}
This work is supported by NSF grants CCF-2103483, IIS-2227669, NSF CAREER Awards CCF-2238358 and CCF-2339310, the UCR Regents Faculty Development Award, and the Google Research Scholar Program.
}

\bibliographystyle{ACM-Reference-Format}
\balance
\bibliography{bib/strings, bib/main, local}


\end{document}
\endinput